\documentclass[journal,twocolumn,10pt]{IEEEtranTCOM}
\usepackage{color}
\usepackage{amsfonts,amsmath,amssymb,amsthm}
\usepackage{latexsym,amscd}
\usepackage{amsbsy}
\usepackage{graphicx,multirow,bm}
\usepackage{algorithm}
\usepackage{algorithmicx}
\usepackage{algpseudocode}
\usepackage{xcolor}
\usepackage{tikz}
\usepackage{diagbox}
\usepackage{multicol}
\usepackage{arydshln}
\usepackage{booktabs}
\usepackage{times}
\usepackage{bm}
\usepackage{amsmath}
\usepackage{cases}
\usepackage{subfigure}
\usepackage{hyperref}
\usepackage{breakurl}

\newtheorem{thm}{Theorem}
\newtheorem{lem}{Lemma}
\newtheorem{corollary}{Corollary}

\newtheorem{con}{Construction}
\newtheorem{defn}{Definition}
\newtheorem{ex}{Example}
\newtheorem{remark}{Remark}

\newcommand{\f}{{\mathbb F}}

\newcommand{\fq}{{\mathbb F}_{q}}
\newcommand{\fqt}{{\mathbb F}_{q^t}}

\newcommand{\tr}{{\rm {Tr}}}
\newcommand{\e}{{\epsilon}}

\renewcommand\thetable{\Roman{table}}

\usepackage{pgf}
\usepackage{tikz}
\usetikzlibrary{matrix,chains,positioning,arrows,calc,decorations,plotmarks,patterns,trees,fit,backgrounds,shapes.multipart,topaths}
\usepackage{pgfplots}
\usetikzlibrary{external}                       
\usepackage{scalefnt}
\usepackage{tkz-graph}
\pgfplotsset{compat=1.3}
\tikzstyle{help lines}=[black!20,dashed]

\pgfplotscreateplotcyclelist{mylist}{red,blue,black,yellow,brown}

\pgfdeclarepatternformonly{my north east lines}{\pgfqpoint{-1pt}{-1pt}}{\pgfqpoint{8pt}{8pt}}{\pgfqpoint{6pt}{6pt}}%
{
	\pgfsetlinewidth{0.4pt}
	\pgfpathmoveto{\pgfqpoint{0pt}{0pt}}
	\pgfpathlineto{\pgfqpoint{6.1pt}{6.1pt}}
	\pgfusepath{stroke}
}
\definecolor{light_gray}{rgb}{0.6,0.6,0.6}
\definecolor{awgray}{rgb}{0.7,0.7,0.7}
\definecolor{awgray_dark}{rgb} {0.4,0.4,0.4}

\tikzset{
	>=stealth',
	mycircle/.style={circle, draw=gray, very thick},
	mycircle_small/.style={circle,draw=awgray_dark,fill = awgray_dark, inner sep=0,minimum size=.6em},
	mycircle_small_black/.style={circle,draw=black,fill = black, inner sep=0,minimum size=.6em},
	mybox/.style={rectangle,rounded corners,draw=black, thick,text width=1em,minimum height=4em,minimum width=5em,text centered},
	mybox_big/.style={rectangle,rounded corners,draw=black, thick,text width=17.5em,minimum height=12em,minimum width=17.5em,text centered},
	mybox_vec/.style={rectangle,rounded corners,draw=black, thick,text width=4em,minimum height=0.7em, minimum width=4em,text centered},
	mybox_vec_short/.style={rectangle,rounded corners,draw=black, thick,text width=1em,minimum height=0.7em, minimum width=2em,text centered},
	pil/.style={->, thick, shorten <=2pt, shorten >=2pt,},
}

\usetikzlibrary{arrows,shapes,chains}
\renewcommand\baselinestretch{1}

\begin{document}
	
	\title{A Transformation of Repairing Reed-Solomon Codes from Rack-Aware Storage Model to Homogeneous Storage Model}
	\author{Yumeng Yang,~\IEEEmembership{Graduate Student Member,~IEEE}, Han Cai,~\IEEEmembership{Member,~IEEE}, and Xiaohu Tang, \IEEEmembership{Senior Member, IEEE}
		\thanks{
			Y. Yang, H. Cai, and X. Tang are with the Information Security and National Computing Grid Laboratory, Southwest Jiaotong University, Chengdu, China (email: yangyumeng@my.swjtu.edu.cn, hancai@swjtu.edu.cn, xhutang@swjtu.edu.cn).
}
}

	\maketitle
\vspace{-2cm}
\begin{abstract}
In this paper, we address the node repair problem of Reed-Solomon (RS) coded distributed storage systems. Specifically,  to overcome the challenges of multiple-node failures of RS codes under the rack-aware storage model, we employ good polynomials to guide the placement of the conventional RS codes into racks and then propose a novel repair framework for the resultant rack-aware RS codes, which can transform its repair to that under the homogeneous storage model. As applications of our repair framework,  firstly {we present the repair scheme of multiple-node failures for some existing constructions, which only have non-trivial solutions for repairing a single-node failure before.} Secondly, we deduce several new constructions of rack-aware RS codes supporting the repair of multiple-node failures {within a single rack and across multiple racks respectively}.

\end{abstract}
	

	\section{Introduction}

{\it Maximum Distance Separable} (MDS) codes, particularly RS codes, have been widely employed as a traditional solution in practical storage systems like Facebook f4 \cite{FB}, Google Colossus \cite{Google}, and Hadoop Distributed File System~\cite{HADOOP}.
 In {MDS coding}, a file is split into $k$ data blocks,  each represented by a symbol of a finite field. In distributed storage systems, node failure is a common problem. To tolerate node failures, the $k$ {symbols} are encoded into $n>k$  {symbols} and distributed across different nodes. If a failure occurs in the system, the conventional solution is to download $k$ symbols from any $k$ surviving nodes, thereby repairing the failed symbol. However, this naive repair scheme is inefficient due to the excessive information downloaded, in fact, the whole file.

{To improve the efficiency of node repair, Dimakis \cite{DGW+10} {\it et al.} initialized the study of regenerating codes that can optimize the repair bandwidth for a given storage system. Subsequently, various MDS code constructions with the optimal repair bandwidth have been proposed \cite{LTT18,RSK11,LZ22,Y20,YB16,YB17,HLZ16,ZZ23,CB19}.
Inspired by these works,}
Guruswami and Wootters
\cite{GW17} proposed a linear repair scheme for a single failure {of RS codes} that significantly reduces the repair bandwidth. Particularly, they introduced the trace function collection technique to repair RS codes, allowing a smaller subfield symbol of each helper node to be sufficient for recovering a single erased symbol.  This seminal idea has sparked significant interest in the repair problem of RS codes, leading to extensive subsequent research for single-node recovery \cite{CV21,BBD+22,TY17,YB16,DDK+17,LWJ19}. On one hand, considerable attention has been paid to designing non-trivial repair schemes for RS codes satisfying explicit parameter ranges with a reasonable field size \cite{BBD+22, XZ23,DDK+17,GW17,LWJ19, ZZ19}. On the other hand, other research has aimed at constructing RS codes that {approach} the optimal repair bandwidth \cite{CV21,TY17,YB16,LWJ19}, nevertheless which normally require a huge field size.  Additionally, in a practical setting,
 the repair mechanism may work when the total amount of failed nodes reaches a given threshold. Therefore, the case of multiple failures is interesting from a practical viewpoint. For multiple-node repair of RS codes, in general, there are two models called {\it centralized model} and {\it cooperative model}. In the centralized model, a repair center is responsible for the repair of all failed nodes \cite{TY19,MBW18,LWJ19}. While in the cooperative model,  in addition to downloading data from helper nodes, the replacement nodes are allowed to collaboratively exchange their individual information {\cite{XZ23,DDK+18,DDK+21,ZZ19}.

{For the repair problem of RS codes, most research has primarily concentrated on the {\it homogeneous storage model}, in which nodes are distributed uniformly in different locations. However, in modern data centers, storage nodes are often distributed in hierarchical topologies. The model in which storage nodes are organized into several equally sized groups and stored in different racks is referred to as {\it rack-aware storage model}. In this model, nodes within the same rack have a significantly lower transmission cost compared to the inter-rack transmission. As a result,  the intra-rack transmission is naturally considered to be free without counting into the repair bandwidth.} The repair of RS codes based on the rack-oriented model has been discussed in {\cite{WC,WZL+23,JLX19, CB19}. In \cite{CB19}, Chen and Barg modified the code family in~\cite{TY17} for the homogeneous storage model to propose rack-oriented codes achieving the optimal bandwidth for repairing a single failure, whose alphabet size is exponential with the code length. In \cite{JLX19}, Jin {\it et al.} introduced the degree decent method and good polynomials into the design of rack-aware RS codes, enabling a linear alphabet size in the code length.  Subsequently, in \cite{WC,WZL+23},  the authors extended the construction \cite{CB19} to repair multiple node failures within a single rack.
}

In this work, we focus on the repair problem of RS codes with multiple failures under the rack-aware storage model. Our contributions are summarized below:
\begin{itemize}
\item A novel repair framework for rack-based RS codes is introduced by leveraging the technique of good polynomials to strategically place nodes within racks. By employing this approach, the multiple-node repair problem of a rack-aware RS code is reduced to the task of  respectively repairing multiple codewords of an RS code under the homogeneous storage model.\\
\item
{ A multiple-node repair scheme is provided for the existing rack-aware RS codes in \cite{JLX19,CB19} based on our repair framework. Unlike the previous approaches in \cite{JLX19,CB19}, which only accommodated bandwidth-efficient repair of a single failed node, our framework demonstrates the capability to handle multiple failures within one rack and even lose the entire rack.}\\
\item By employing different good polynomials, a series of new rack-aware RS codes are constructed, which have good repair properties {for single-rack or multiple-rack failure} {provided that the number of failures does not exceed the recoverability by the minimum Hamming distance.}\\
\end{itemize}

The remainder of this paper is organized as follows. In Section \ref{sec: pre}, we introduce some basic concepts and recall the trace repair framework of RS codes under the homogeneous model. In Section \ref{sec: framework}, we present a repair framework for multiple-node recovery for rack-aware RS codes. In Section \ref{sec: construction}, we interpret and extend the existing constructions by our repair framework to support multiple-node failures in a single rack. Moreover, we employ three classes of good polynomials to design new rack-aware RS codes, {supporting single-rack or multiple-rack recovery}. Finally, the conclusion is given in Section \ref{sec: conclusion}.

\section{Preliminaries}\label{sec: pre}

In this section,  we review linear repair schemes of RS codes under the homogeneous storage model and rack-aware storage model, respectively.
\subsection{Definitions and Notations}

For any positive integers $a<b$, denote by $[a]$ the set $\{1,2,\cdots,a\}$ and $[a,b]$ the set $\{a,\cdots,b\}$. Let $q$ be a prime power and $\fq$ be a finite field of $q$ elements.
Denote by $\fqt$ the field extension of $\fq$ with $[\fqt:\fq]=t$.
Let $\fqt[x]$ be the ring of polynomials over $\fqt$.
The trace function from $\fqt$ to $\fq$ is a polynomial defined by
$$\tr_{\fqt/\fq}(x)=x+x^q+\cdots+x^{q^{t-1}},$$
which is an $\fq$-linear function. For simplicity, we denote $\tr_{\fqt/ \fq}$ as $\tr$  except otherwise specified.

Let $\{\beta_{1},\beta_{2},\cdots,\beta_{t}\}$ be a basis of $\fqt$ over $\fq$ and $\{\beta^{\bot}_{1},\beta^{\bot}_{2},\cdots,\beta^{\bot}_{t}\}$ be its dual basis  satisfying
$$\tr(\beta_{i}\beta^{\bot}_{j})=
\begin{cases}
0,&i \neq j;\\
1,&i=j.
\end{cases}
$$
Then, any $\alpha\in\fqt$ can be represented by $\{\beta^{\bot}_{1},\cdots,\beta^{\bot}_{t}\}$, whose coefficients are uniquely determined  by $\tr(\beta_{i}\alpha)$, $i\in[t]$,
i.e.,
\begin{equation}\label{eq:dual basis}
\alpha=\sum_{i=1}^{t}\tr(\beta_{i}\alpha)\beta_{i}^{\bot}.
\end{equation}

Let $U$ be a subspace of $\fqt$ with the dimension $s$ over $\fq$, then a linearized polynomial $L_{U}(x)$ over $\fqt$ is defined by
 $$L_{U}(x)=\Pi_{u\in U}(x-u).$$
Clearly, the trace function is a special case of the linearized polynomial when $U$ is a subfield of $\fqt$.

\begin{defn}[Reed-Solomon Code]\label{Def_RS}
Let $A=\{\alpha_{1},\cdots,\alpha_{n}\}$ be the set of distinct elements over $\fqt$. A Reed-Solomon code of length $n$ and dimension $k$ with evaluation points $A$ is defined as
\begin{equation*}
\begin{split}
{\rm RS}(n,k,A)=&\{(f(\alpha_{1}),f(\alpha_{2}),\cdots,f(\alpha_{n})):\\
&\hspace{2cm} f\in\fqt[x],deg(f)<k\}.
\end{split}
\end{equation*}
\end{defn}

Similarly, a Generalized Reed-Solomon code ${\rm GRS}(n,k, A, \boldsymbol{\nu})$  can be  defined by the set of vectors $$\{(\nu_{1}f(\alpha_{1}),\cdots,\nu_{n}f(\alpha_{n})):f\in\fqt[x],deg(f)<k\},$$ where $ \boldsymbol{\nu}=(\nu_{1},\cdots,\nu_{n})\in(\fqt^{*})^n$.

\begin{defn}[Dual Code]\label{def:dual}
Let $\mathcal{C}$ be a linear code over $\fqt$ of length $n$. The dual code $\mathcal{C}^{\bot}$ of  code $\mathcal{C}$  is defined as
$$\mathcal{C}^{\bot}=\{(c'_{1},c'_{2},\cdots,c'_{n}):\sum_{i=1}^{n}c'_{i}c_{i}=0,\,\,\,\,\forall (c_{1},\cdots,c_{n})\in\mathcal{C}\}.$$
\end{defn}

 Precisely, $({\rm RS}(n,k,A))^{\bot}={\rm GRS}(n,n-k, A, \boldsymbol{\nu})$, where $\nu_{i}=\Pi_{j\neq i}(\alpha_{i}-\alpha_{j})^{-1}$, $i\in[n]$. For simplicity, from now on we omit the multipliers $\nu_{i}$ for $i\in[n]$ in the involved GRS code because it does not affect our subsequent discussion.

\subsection{Repairing Reed-Solomon Codes under the Homogeneous Storage Model}\label{subsec: RS}

Assume that a file $M$ is divided into $k$ data blocks and encoded into a codeword using an RS code of length $n$ over $\fqt$. Then, the $n$ symbols of the codeword are distributed across $n$ independent storage nodes. Since each symbol is represented by an element over $\fqt$, the contents of a node can then be regarded as $t$ sub-symbols over $\fq$.

 \subsubsection{Single Failure} When one single node fails, a new replacement node contacts $d\,(k\leq d\leq n-1)$ surviving nodes and downloads $b_{i}$ sub-symbols from each helper node $i\in[d]$ to recover the desired symbol at the failed node.  The repair bandwidth $b$ is defined as the amount of all the sub-symbols downloaded by the repair center, i.e., $b=\sum_{i=1}^{d}b_{i}$.

In the literature~\cite{GW17}, Guruswami and Wootters gave an exact characterization of linear repair schemes for  RS code with a single failure and $d=n-1$ helper nodes.
{It is well known that
any polynomial $f(x)\in\mathbb{F}_{q^t}[x]$ of degree less than $k$ corresponds to a codeword of RS code $\mathcal{C}={\rm RS}(n,k,A)$ over the finite field $\mathbb{F}_{q^t}$.}
Suppose that the symbol $f(\alpha_{i^*})$ is failed, where $\alpha_{i^*}\in A$. The repair procedure of Guruswami and Wootters' scheme proceeds as follows.

\textbf{Step~1:} Find $t$ polynomials  $\{g_{1}(x),\cdots,g_{t}(x)\}\in\fqt[x]$ corresponding to $t$ codewords of $\mathcal{C}^{\bot}$ such that $\{g_{l}(\alpha_{i^*}):l\in[t]\}$ forms a basis of $\fqt$ over $\fq$. Assume that $g_{l}(\alpha_{i^*})=\beta_{l}$ for $l\in[t]$. Then, one can deduce the following parity-check equations:
\begin{eqnarray}\label{Eqn_PCE}
g_{l}(\alpha_{i^*})f(\alpha_{i^*})=\beta_{l}f(\alpha_{i^*})=-\sum_{i\neq i^*}g_{l}(\alpha_{i})f(\alpha_{i}), \,\,l\in[t].
\end{eqnarray}

\textbf{Step~2:}  By applying the trace function to \eqref{Eqn_PCE},  obtain repair equations
\begin{eqnarray}\label{Eqn_Sub_Trace}
\tr(g_{l}(\alpha_{i^*})f(\alpha_{i^*}))=-\sum_{i\neq i^*}\tr(g_{l}(\alpha_{i})f(\alpha_{i}))
\end{eqnarray}
for all $l\in[t]$. Thus,  node $i^*$   downloads sub-symbols $\{\tr(\gamma_{i,w}f(\alpha_{i})):i\in[n]\backslash{\{i^*\}};w\in[b_{i}]\}$  from the remaining $n-1$  nodes, where $\{\gamma_{i,w}:w\in[b_{i}]\}$ is a basis of ${\rm Span}_{\fq}(g_{l}(\alpha_{i}):l\in[t])$ over $\fq$ and $b_i={\rm Rank}_{\fq}\{g_{l}(\alpha_{i}):l\in [t]\}$.

\textbf{Step~3:} {
By means of $\{\mathrm{Tr}(\beta_{l}f(\alpha_{i^*})):l\in[t]\}$ obtained from \eqref{Eqn_Sub_Trace}, node $i^*$  recovers the failed symbol}
\begin{eqnarray*}
f(\alpha_{i^*})=\sum_{i=1}^{t}\tr(\beta_{i}f(\alpha_{i^*}))\beta_{i}^{\bot}
\end{eqnarray*}
according to \eqref{eq:dual basis}.

Obviously, node $i^*$ {needs to} download  $b_i={\rm Rank}_{\fq}\{g_{l}(\alpha_{i}):l\in [t]\}$ sub-symbols from node $i\in [n]\setminus\{i^*\}$. Therefore, the repair bandwidth is
\begin{eqnarray}\label{Eqn_RB_WGS}
b=\sum_{i\in [n]\setminus\{i^*\}}b_i=\sum_{i\in [n]\setminus\{i^*\}}{\rm Rank}_{\fq}\{g_{l}(\alpha_{i}):l\in [t]\}.
\end{eqnarray}

In the above approach, the repair of an RS code is converted to searching for suitable dual codewords to minimize the repair bandwidth $b$ in \eqref{Eqn_RB_WGS}. This seminal work led to the efficient single-node repair schemes of RS codes
in~\cite{BBD+22,TY17,YB16,DDK+17,LWJ19}.


\subsubsection{ Multiple Failures}\label{sec: multiple}
The repair of multiple failures is divided into two cases {\it centralized repair} and {\it cooperative repair}, according to
whether there exists a repair center.
Roughly speaking,   when $\e\,(1<\e\leq n-k)$ nodes fail,
the centralized repair model assigns a repair center responsible for downloading all necessary data from the helper nodes and then repairing all failed nodes. In contrast, under the cooperative repair model, each replacement node firstly independently downloads data from helper nodes and next exchanges data with each other. It should be noted that concerning the repair bandwidth, the former model only counts the amount of information downloaded by the repair center, whereas the latter considers the total information transmission, including both downloaded and exchanged data.

Indeed, there are few multiple-node repair schemes of RS codes under these two models,  which are all generalized from Guruswami and Wootters'  trace repair framework. Currently,  {only a few explicit constructions for the centralized repair model exist, e.g., \cite{TY19, MBW18}}, while the known ones for the cooperative repair model require that the number of failed nodes is at most $3$ {\cite{XZ23, DDK+18,DDK+21,ZZ19}.}

\subsection{Repairing Reed-Solomon Codes under the  Rack-Aware Storage Model}

Under the rack-aware model, the $n$ storage nodes are equally organized into $\bar{n}$ groups, i.e., $\bar{n}|n$. A group is referred to as a rack, each containing $u=\frac{n}{\bar{n}}$ storage nodes. {Denote by $m$ the number of failed racks satisfying $1\leq m\leq \bar{n}-\lceil\frac{k}{u}\rceil$. Each rack contains a relayer, which can be seen as an internal repair center responsible for transferring and processing data within this rack. Unlike traditional settings of the rack-aware storage model, this relayer is a virtual node without storage functionality, which is more consistent with the realistic rack architecture. }

Once a rack occurs  $1\leq \e\leq u$ node failures, a repair operation is performed below.
Let $(c_{i^*,1},\cdots,c_{i^*,u})$ be the symbols stored in $i^*$-th rack. Without loss of generality, assume that  the $\e$ symbols $c_{i^*,1},\cdots,c_{i^*,\e}$  at  $\e$ nodes fail respectively. Then, {the relayer in rack $i^*$ connects to  the relayer in} each of {$\lceil\frac{k}{u}\rceil\leq \bar{d}\leq \bar{n}-1$}  helper racks, to request $\bar{b}_{1},\bar{b}_{2},\cdots,\bar{b}_{\bar{d}}$ sub-symbols respectively. Next, each {relayer node} accesses all the symbols within the same rack to form some linear combinations as the desired sub-symbols.  After that, the repair center recovers the failed  $\e$  symbols with the help of surviving $u-\e$ nodes within rack $i^*$. Note that, unlike the homogeneous storage model,  this model is characterized by heterogeneousness between intra-rack and inter-rack. More precisely, the nodes inside the same rack share the storage information, {and the communication between the relayer} and the nodes within this rack is assumed to be free. Hence, in this system,} only the cross-rack transmission is taken into account of the repair bandwidth, i.e.,
\begin{eqnarray}\label{Eqn_Rack_RB}
b=\sum_{i=1}^{\bar{d}}{\bar{b}}_{i}
\end{eqnarray}
which results in a significant benefit compared to the homogeneous coding.

Furthermore, repair models involving multiple-rack failures {are similar to} Section~\ref{sec: multiple}. In centralized repair, a {repair center independent of the racks} downloads data from {$\lceil\frac{k}{u}\rceil\leq\bar{d}\leq \bar{n}-m$ relayers} of helper racks instead of individual helper nodes. After centrally repairing the failed nodes, the repaired symbols are transmitted back to each failed rack. In collaborative repair, racks containing the failed nodes individually download data from helper racks and then collaboratively repair the failed racks through information interaction among their {individual relayers.}

 In~\cite{JLX19}, Jin {\it et al.}  generalized the repair framework in~\cite{GW17} to the single-node recovery of RS codes under the rack-aware storage model.

 \section{A Repair Framework of Rack-Aware Reed-Solomon Codes}\label{sec: framework}

In this section, we transform the repair of a rack-oriented RS code into a conventional repair problem of an RS code under the homogeneous storage model.

\subsection{Rack-Aware Reed-Solomon Codes via Good Polynomials}\label{sec:repair model}

We begin with the formal definition of  RS codes under the rack-aware storage model.

\begin{defn}[Rack-Aware Reed-Solomon Code] Let $u$ be the size of rack and $A=\{\alpha_{i,j}:i\in[\bar{n}];\,j\in[u]\}$ be the set of distinct elements over $\f_{q^t}$. A rack-aware Reed-Solomon code of length $n=\bar{n}u$ and dimension $k=\bar{k}u+v\,\,(0\leq v\leq u-1)$  with evaluation points $A$ is defined as
\begin{equation*}
\begin{split}
{\rm{RS}_{rack}}(n,k,A,u)=&\{(f(\alpha_{1,1}),f(\alpha_{1,2}),\cdots,
f(\alpha_{\bar{n},1}),\cdots,\\
&\quad f(\alpha_{\bar{n},u})):
f\in\fqt[x],deg(f)<k\},
\end{split}
\end{equation*}
such that the nodes in $i$-th rack can be represented as
$$(f(\alpha_{i,1}),f(\alpha_{i,2}),\cdots,f(\alpha_{i,u})),\,\,\,i\in[\bar{n}].$$
\end{defn}

Hereafter, we construct rack-aware RS codes by applying good polynomials to RS codes.

 \begin{defn}[Good Polynomial \cite{TB14}]  A polynomial $h(x)\in \fqt[x]$ of degree $u$ is called a good polynomial if there exists a partition $A=\cup_{i=1}^{\bar{n}}A_{i}$ over $\fqt$ of size $n$ with $|A_{i}|=u$, such that $h(x)$ remains constant on each set $A_{i}$. In other words, $h(\alpha)=y_{i}$ for any $\alpha\in A_{i}$, where $y_{i}$ is a constant in $\fqt$ for any $i\in[\bar{n}]$.
\end{defn}

 {
   \begin{con}\label{cons}
  Let $n$, $k$, $u$, $\bar{n}$ be positive integers with $u\mid n$ and $\bar{n}=\frac{n}{u}$.
  Let $f(x)\in \fqt[x]$ be a polynomial with degree at most $k-1$ and $h(x)$ be a good polynomial over $A=\cup_{i=1}^{\bar{n}}A_{i}$
  where $A_i=\{\alpha_{i,1},\alpha_{i,2},\cdots,\alpha_{i,u}\}$.
  Then, based on  the good polynomial $h(x)$, the conventional RS codeword $(f(x): x\in {A})$ can be transformed into a rack-aware RS  codeword  in the array form of
  \begin{eqnarray}\label{eqn_cons_1}
\begin{pmatrix}
  f(\alpha_{1,1})&f(\alpha_{1,2})&\cdots&f(\alpha_{1,u})\\
  f(\alpha_{2,1})&f(\alpha_{2,2})&\cdots&f(\alpha_{2,u})\\
  \vdots&\vdots&\cdots&\vdots\\
  f(\alpha_{\bar{n},1})&f(\alpha_{\bar{n},2})&\cdots&f(\alpha_{\bar{n},u})\\
  \end{pmatrix},
  \end{eqnarray}
where the $i$-row corresponds to the $i$-th rack for  $i\in [\bar{n}]$.

Given  $i\in [\bar{n}]$, denote $h(\alpha_{i,j})=y_i$ and define $f_i(x)\equiv f(x)\pmod{h(x)-y_i}$ with $\deg{(f_i)}<u$, which
implies
\begin{eqnarray*}
f_i(\alpha_{i,j})=f(\alpha_{i,j}),\,j\in[u].
\end{eqnarray*}
Thus, \eqref{eqn_cons_1} can be rewritten as
 \begin{eqnarray}\label{eqn_cons}
\begin{pmatrix}
  f_1(\alpha_{1,1})&f_1(\alpha_{1,2})&\cdots&f_1(\alpha_{1,u})\\
  f_2(\alpha_{2,1})&f_2(\alpha_{2,2})&\cdots&f_2(\alpha_{2,u})\\
  \vdots&\vdots&\cdots&\vdots\\
  f_{\bar{n}}(\alpha_{\bar{n},1})&f_{\bar{n}}(\alpha_{\bar{n},2})&\cdots&f_{\bar{n}}(\alpha_{\bar{n},u})\\
  \end{pmatrix},
  \end{eqnarray}
  \end{con}

\begin{remark}
For any $i\in[\bar{n}]$, the relayer of rack $i$ can obtain the residue polynomial $f_i(x)$ from the $u$ pairs $(\alpha_{i,j},f(\alpha_{i,j})),\,j\in[u]$,
by means of the Lagrange interpolation formula, i.e.,
\begin{equation*}
f_i(x) = \sum_{j=1}^{u} f(\alpha_{i,j}) \prod_{\substack{w=1 \\ w \neq j}}^{u} \frac{x - \alpha_{i,w}}{\alpha_{i,j} - \alpha_{i,w}}.
\end{equation*}
\end{remark}
 }

In the following, we introduce a useful lemma for
the repair of the rack-aware RS  codes in \eqref{eqn_cons}  \cite{CMST2023}.

\begin{lem}\label{cross-rack rs}
Let $f(x)$ and $h(x)$ respectively be two polynomials of degree $k-1$ and $u$ over $\f_{q^t}$. Set $s=\lceil {k\over u}\rceil$.
For any $y\in \f_{q^t}$, denote the residue polynomial $f(x)\bmod\, (h(x)-y)$ by
\begin{eqnarray*}
\sum_{j=0}^{u-1}H_{j}(y)x^{j}\equiv f(x)\bmod\, (h(x)-y).
\end{eqnarray*}
Then, $H_{j}(y)$ is a polynomial of $deg(H_{j}(y))\leq s-1$.
\end{lem}

\begin{proof} Herein we apply the Euclidean algorithm to $f(x)$ as follows:
\begin{eqnarray*}
f(x)&=& h(x)u_1(x)+v_1(x), \\
u_1(x)&=& h(x)u_2(x)+v_2(x), \\
&\vdots&\\
u_{s-1}(x)&=& h(x)u_{s}(x)+v_{s}(x),
\end{eqnarray*}
where the polynomials $u_1(x),v_1(x),\cdots,u_{s}(x),v_{s}(x)$ satisfy $0\le deg(v_1(x)),\cdots, deg(v_{s}(x))<u$ and
$u_{s }(x)=0$. Therefore,
\begin{equation*}
\begin{split}
f(x)=&h^{s-1}(x)v_{s}(x)+h^{s-2}(x)v_{s-1}(x)+\cdots\\
& +h(x)v_{2}(x)+v_{1}(x).
\end{split}
\end{equation*}
which gives
\begin{equation*}
\begin{split}
&f(x)\bmod\,(h(x)-y)\\
\equiv& y^{s-1}v_{s}(x)+y^{s-2}v_{s-1}(x)+\cdots+yv_{2}(x)+v_{1}(x).
\end{split}
\end{equation*}

This is to say, we can arrange the polynomial in the form of
\begin{equation*}
\begin{split}
&f(x)\bmod\,(h(x)-y)\\
\equiv& H_{u-1}(y)x^{u-1}+H_{u-2}(y)x^{u-2}+\cdots+H_{0}(y),
\end{split}
\end{equation*}
where $deg(H_{j}(x))\leq s-1$. This completes the proof.
\end{proof}

 \subsection{Repair of Rack-Aware Reed-Solomon Codes}\label{sec: good polynomial}
In this subsection, we propose a repair framework for rack-aware RS codes given in \eqref{eqn_cons} capable of handling multiple failures.
To this end, we utilize {{good polynomials}} to
transform the problem of multiple-node recovery into the recovery of a homogeneous RS code.

First of all, we prove the following result by  Lemma~\ref{cross-rack rs}, which is crucial for our transformation.

\begin{thm}\label{thm:rack_rs}
Let $y_1,\cdots, y_{\bar{n}}$ be $\bar{n}$ distinct elements over $\f_{q^t}$ such that
$h(\alpha)=y_i$ for $\alpha\in A_i$, where $\deg(h(x))=u$ and $|A_i|=u$ for $1\leq i\leq \bar{n}$. Suppose that $deg(f(x))<k$ and
the residue polynomial $f_{i}(x)\equiv f(x)\bmod\,(h(x)-y_i)$ is of the form
\begin{eqnarray}\label{Eqn_f_ix}
f_{i}(x)=\sum_{j=0}^{u-1}e_{i,j}x^j,\,1\leq i\leq \bar{n}.
\end{eqnarray}
Then,  $(e_{1,j},e_{2,j},\cdots,e_{\bar{n},j})$
forms a codeword of an {$[\bar{n},\leq \lceil\frac{k}{u}\rceil]$} RS code {with evaluation points $\bar{A}=\{y_{i}:i\in[\bar{n}]\}$}  for
any $j\in\{0,\cdots,u-1\}$.
\end{thm}

\begin{proof}
By Lemma \ref{cross-rack rs}, we have
\begin{eqnarray*}
f_i(x)=\sum_{j=0}^{u-1}H_{j}(y_i)x^{j},
\end{eqnarray*}
together with \eqref{Eqn_f_ix} which implies
\begin{eqnarray*}
e_{i,j}=H_j(y_i).
\end{eqnarray*}

Recall that $k=\bar{k}u+v, 0\leq v\leq u-1$. It then follows from Lemma \ref{cross-rack rs} that $H_j(y)$ is a polynomial of degree less than $\lceil {k\over u}\rceil$.
Thus, $(e_{1,j},e_{2,j},\cdots,e_{\bar{n},j})$
forms a codeword of an {RS code of dimension at most $\lceil {k\over u}\rceil$} for
any $j\in\{0,\cdots,u-1\}$ by Definition \ref{Def_RS}.
\end{proof}

Notably, Theorem \ref{thm:rack_rs} bridges the repair of RS codes under the homogeneous storage model and that under the rack-aware storage model. We briefly illustrate the main idea in Fig. $1$.  It follows from \eqref{Eqn_f_ix} that for any $i\in [\bar{n}]$,
acquiring the data $(f(\alpha_{i,1}),\cdots,f(\alpha_{i,u}))$ is equivalent to obtaining  the coefficients $(e_{i,0},\cdots,e_{i,u-1})$ of $f_{i}(x)$. {Assume that $\epsilon_{i^*}$ failures occur in rack $i^*$. Let  $ \mathcal{E}_{i^*}\subseteq [0,u-1]$ with $| \mathcal{E}_{i^*}|=\epsilon_{i^*}$ and  set $\mathcal{M}_{i^*}=[0,u-1]\setminus \mathcal{E}_{i^*}$.
Then, in order to repair rack $i^*$, it is sufficient to recover $\{e_{i^* ,j} : j\in  \mathcal{E}_{i^*}\}$ since after that we can compute the
 remaining $u-\epsilon_{i^*}$ coefficients $\{e_{i^* ,j} : j\in  \mathcal{M}_{i^*}\}$ within
rack $i^*$ if the $u-\epsilon_{i^*}$ elements in $\mathcal{M}_{i^*}$ are consecutive   (See  Step 3 of the repair procedure below). Without loss of generality, we compute the first
$\epsilon_{i^*}$ coefficients $\{e_{i^* ,u-j}: j\in[\epsilon_{i^*} ]\}$ as an example, i,e, $\mathcal{E}_{i^*}=[u-\epsilon_{i^*},u-1]$ and  $\mathcal{M}_{i^*}=[0,u-\epsilon_{i^*}-1]$.}  Again according to  Theorem \ref{thm:rack_rs}, $(e_{1,u-j}, \cdots, e_{\bar{n},u-j})$ forms an {RS codeword of dimension at most $\lceil\frac{k}{u}\rceil$  for any $j\in[\e_{i^*}]$. {Note that an $[\bar{n},\leq \lceil\frac{k}{u}\rceil]$ RS code can be regarded as a subcode of an $[\bar{n},\lceil \frac{k}{u}\rceil ]$ RS code when they share the same evaluation points.
 Thus, we can repair it as an $[\bar{n},\lceil\frac{k}{u}\rceil]$ RS code.}  As a result, we can get the desired coefficients  $\{e_{i^*,u-j}:j\in[\e_{i^*}]\}$  by applying some known repair schemes under the homogeneous storage model to such $\e_{i^*}$ codewords respectively.
That is,  the repair of a rack-oriented RS code is therefore transformed into a traditional repair of an RS code under the homogeneous storage model.

{

\begin{figure*}[ht]\label{fig_transform}
\centering

\tikzset{every picture/.style={line width=0.75pt}} 

\begin{tikzpicture}[x=0.75pt,y=0.75pt,yscale=-1,xscale=1]

\draw   (86.11,276.57) -- (292.58,276.57) -- (292.58,304.19) -- (86.11,304.19) -- cycle ;
\draw   (583.76,278.23) -- (613.33,278.23) -- (613.33,305.52) -- (583.76,305.52) -- cycle ;
\draw   (583.76,334.16) -- (613.33,334.16) -- (613.33,361.45) -- (583.76,361.45) -- cycle ;
\draw   (583,416.63) -- (613.33,416.63) -- (613.33,444) -- (583,444) -- cycle ;
\draw   (416.95,278.23) -- (449.07,278.23) -- (449.07,305.52) -- (416.95,305.52) -- cycle ;
\draw   (417.12,334.17) -- (450.07,334.17) -- (450.07,361.45) -- (417.12,361.45) -- cycle ;
\draw   (416.12,416.63) -- (451.07,416.63) -- (451.07,443.92) -- (416.12,443.92) -- cycle ;
\draw  [color={rgb, 255:red, 11; green, 33; blue, 244 } ,draw opacity=1 ][dash pattern={on 3.75pt off 3pt on 7.5pt off 1.5pt}]  (572.03,270.03) .. controls (572.03,264.41) and (576.58,259.86) .. (582.2,259.86) -- (612.7,259.86) .. controls (618.32,259.86) and (622.87,264.41) .. (622.87,270.03) -- (622.87,450.93) .. controls (622.87,456.54) and (618.32,461.1) .. (612.7,461.1) -- (582.2,461.1) .. controls (576.58,461.1) and (572.03,456.54) .. (572.03,450.93) -- cycle ;
\draw  [color={rgb, 255:red, 11; green, 33; blue, 244 }  ,draw opacity=1 ][dash pattern={on 3.75pt off 3pt on 7.5pt off 1.5pt}] (408.22,270.67) .. controls (408.22,265.05) and (412.77,260.5) .. (418.39,260.5) -- (448.89,260.5) .. controls (454.5,260.5) and (459.06,265.05) .. (459.06,270.67) -- (459.06,454.93) .. controls (459.06,460.54) and (454.5,465.09) .. (448.89,465.09) -- (418.39,465.09) .. controls (412.77,465.09) and (408.22,460.54) .. (408.22,454.93) -- cycle ;
\draw    (138.25,276.57) -- (138.25,303.85) ;
\draw    (190.39,276.57) -- (190.39,303.85) ;
\draw    (242.54,276.57) -- (242.54,303.85) ;
\draw   (86.11,331.14) -- (292.58,331.14) -- (292.58,358.88) -- (86.11,358.88) -- cycle ;
\draw    (138.25,331.14) -- (138.25,358.42) ;
\draw    (190.39,331.14) -- (190.39,358.42) ;
\draw    (242.54,331.14) -- (242.54,358.42) ;
\draw   (86.11,415.65) -- (292.58,415.65) -- (292.58,443.27) -- (86.11,443.27) -- cycle ;
\draw    (138.25,415.65) -- (138.25,442.94) ;
\draw    (190.39,415.65) -- (190.39,442.94) ;
\draw    (242.54,415.65) -- (242.54,442.94) ;
\draw [color={rgb, 255:red, 208; green, 2; blue, 27 }  ,draw opacity=1 ]   (138.25,276.57) -- (190.39,303.85) ;
\draw [color={rgb, 255:red, 208; green, 2; blue, 27 }  ,draw opacity=1 ]   (138.25,303.85) -- (190.39,276.57) ;
\draw [color={rgb, 255:red, 208; green, 2; blue, 27 }  ,draw opacity=1 ]   (418.23,305.52) -- (449.07,278.23) ;
\draw [color={rgb, 255:red, 208; green, 2; blue, 27 }  ,draw opacity=1 ]   (416.95,278.23) -- (449.07,305.52) ;
\draw   (480.95,278.23) -- (513.07,278.23) -- (513.07,305.52) -- (480.95,305.52) -- cycle ;
\draw   (481.12,334.17) -- (514.07,334.17) -- (514.07,361.45) -- (481.12,361.45) -- cycle ;
\draw   (480.12,416.63) -- (515.07,416.63) -- (515.07,443.92) -- (480.12,443.92) -- cycle ;
\draw  [color={rgb, 255:red, 11; green, 33; blue, 244 }  ,draw opacity=1 ][dash pattern={on 3.75pt off 3pt on 7.5pt off 1.5pt}] (472.22,270.67) .. controls (472.22,265.05) and (476.77,260.5) .. (482.39,260.5) -- (512.89,260.5) .. controls (518.5,260.5) and (523.06,265.05) .. (523.06,270.67) -- (523.06,454.93) .. controls (523.06,460.54) and (518.5,465.09) .. (512.89,465.09) -- (482.39,465.09) .. controls (476.77,465.09) and (472.22,460.54) .. (472.22,454.93) -- cycle ;
\draw [color={rgb, 255:red, 208; green, 2; blue, 27 }  ,draw opacity=1 ]   (480.95,305.52) -- (513.07,278.23) ;
\draw [color={rgb, 255:red, 208; green, 2; blue, 27 }  ,draw opacity=1 ]   (513.07,305.52) -- (480.95,278.23) ;
\draw  [color={rgb, 255:red, 11; green, 33; blue, 244 }  ,draw opacity=1 ][dash pattern={on 4.5pt off 4.5pt}] (29,301.08) .. controls (29,279.77) and (46.27,262.5) .. (67.58,262.5) -- (266.42,262.5) .. controls (287.73,262.5) and (305,279.77) .. (305,301.08) -- (305,416.83) .. controls (305,438.14) and (287.73,455.41) .. (266.42,455.41) -- (67.58,455.41) .. controls (46.27,455.41) and (29,438.14) .. (29,416.83) -- cycle ;
\draw [color={rgb, 255:red, 208; green, 2; blue, 27 }  ,draw opacity=1 ]   (138.25,358.42) -- (190.39,331.14) ;
\draw [color={rgb, 255:red, 208; green, 2; blue, 27 }  ,draw opacity=1 ]   (138.25,331.14) -- (190.39,358.42) ;
\draw   (294,290.5) -- (305.23,286.5) -- (305.23,289.41) -- (403.77,289.41) -- (403.77,286.5) -- (415,290.5) -- (403.77,294.5) -- (403.77,291.59) -- (305.23,291.59) -- (305.23,294.5) -- cycle ;
\draw [color={rgb, 255:red, 208; green, 2; blue, 27 }  ,draw opacity=1 ]   (418.17,361.45) -- (449.02,334.17) ;
\draw [color={rgb, 255:red, 208; green, 2; blue, 27 }  ,draw opacity=1 ]   (417.95,334.17) -- (450.07,361.45) ;
\draw [color={rgb, 255:red, 208; green, 2; blue, 27 }  ,draw opacity=1 ]   (86.11,304.19) -- (138.25,276.9) ;
\draw [color={rgb, 255:red, 208; green, 2; blue, 27 }  ,draw opacity=1 ]   (86.11,276.9) -- (138.25,304.19) ;
\draw [color={rgb, 255:red, 153; green, 178; blue, 173 }  ,draw opacity=1 ][line width=0.75]  [dash pattern={on 0.84pt off 2.51pt}]  (411.75,278.55) -- (626.4,277.92) ;
\draw [color={rgb, 255:red, 153; green, 178; blue, 173 }  ,draw opacity=1 ][line width=0.75]  [dash pattern={on 0.84pt off 2.51pt}]  (410.75,305.55) -- (625.4,304.92) ;
\draw [color={rgb, 255:red, 153; green, 178; blue, 173 }  ,draw opacity=1 ][line width=0.75]  [dash pattern={on 0.84pt off 2.51pt}]  (406.75,334.48) -- (621.4,333.85) ;
\draw [color={rgb, 255:red, 153; green, 178; blue, 173 }  ,draw opacity=1 ][line width=0.75]  [dash pattern={on 0.84pt off 2.51pt}]  (406.75,361.77) -- (623,362) ;
\draw [color={rgb, 255:red, 153; green, 178; blue, 173 }  ,draw opacity=1 ][line width=0.75]  [dash pattern={on 0.84pt off 2.51pt}]  (407.74,416.95) -- (622.39,416.32) ;
\draw [color={rgb, 255:red, 153; green, 178; blue, 173 }  ,draw opacity=1 ][line width=0.75]  [dash pattern={on 0.84pt off 2.51pt}]  (407.74,444.24) -- (622.39,443.6) ;
\draw   (295,346.5) -- (306.23,342.5) -- (306.23,345.41) -- (404.77,345.41) -- (404.77,342.5) -- (416,346.5) -- (404.77,350.5) -- (404.77,347.59) -- (306.23,347.59) -- (306.23,350.5) -- cycle ;
\draw   (293,429.5) -- (304.23,425.5) -- (304.23,428.41) -- (402.77,428.41) -- (402.77,425.5) -- (414,429.5) -- (402.77,433.5) -- (402.77,430.59) -- (304.23,430.59) -- (304.23,433.5) -- cycle ;

\draw (38.54,283.38) node [anchor=north west][inner sep=0.75pt]  [font=\large] [align=left] {{\footnotesize Rack 1}};
\draw (189.31,373.65) node [anchor=north west][inner sep=0.75pt]  [font=\large,rotate=-90.04]  {$...$};
\draw (37.54,338.22) node [anchor=north west][inner sep=0.75pt]  [font=\large] [align=left] {{\footnotesize Rack 2}};
\draw (37.99,423.1) node [anchor=north west][inner sep=0.75pt]  [font=\large] [align=left] {{\footnotesize Rack $\displaystyle \bar{n}$}};
\draw (322.07,269.41) node [anchor=north west][inner sep=0.75pt]  [font=\small] [align=left] {$\displaystyle h( A_{1}) =y_{1}$};
\draw (60.53,372.64) node [anchor=north west][inner sep=0.75pt]  [font=\large,rotate=-90.04]  {$...$};
\draw (356.08,373.65) node [anchor=north west][inner sep=0.75pt]  [font=\large,rotate=-90.04]  {$...$};
\draw (539.2,357.18) node [anchor=north west][inner sep=0.75pt]  [font=\large,rotate=-359.95]  {$...$};
\draw (89.86,280.43) node [anchor=north west][inner sep=0.75pt]  [font=\footnotesize] [align=left] {$\displaystyle f( \alpha _{1,1})$};
\draw (142.78,280.43) node [anchor=north west][inner sep=0.75pt]  [font=\footnotesize] [align=left] {$\displaystyle f( \alpha _{1,2})$};
\draw (246.61,280.43) node [anchor=north west][inner sep=0.75pt]  [font=\footnotesize] [align=left] {$\displaystyle f( \alpha _{1,u})$};
\draw (322.07,326.3) node [anchor=north west][inner sep=0.75pt]  [font=\small] [align=left] {$\displaystyle h( A_{2}) =y_{2}$};
\draw (322.06,409.14) node [anchor=north west][inner sep=0.75pt]  [font=\small] [align=left] {$\displaystyle h( A_{\bar{n}}) =y_{\bar{n}}$};
\draw (90.86,334.33) node [anchor=north west][inner sep=0.75pt]  [font=\footnotesize] [align=left] {$\displaystyle f( \alpha _{2,1})$};
\draw (90.86,419.16) node [anchor=north west][inner sep=0.75pt]  [font=\footnotesize] [align=left] {$\displaystyle f( \alpha _{\bar{n} ,1})$};
\draw (142.78,335.33) node [anchor=north west][inner sep=0.75pt]  [font=\footnotesize] [align=left] {$\displaystyle f( \alpha _{2,2})$};
\draw (247.61,334.33) node [anchor=north west][inner sep=0.75pt]  [font=\footnotesize] [align=left] {$\displaystyle f( \alpha _{2,u})$};
\draw (142.78,419.16) node [anchor=north west][inner sep=0.75pt]  [font=\footnotesize] [align=left] {$\displaystyle f( \alpha _{\bar{n} ,2})$};
\draw (245.61,420.16) node [anchor=north west][inner sep=0.75pt]  [font=\footnotesize] [align=left] {$\displaystyle f( \alpha _{\bar{n} ,u})$};
\draw (206.73,290.31) node [anchor=north west][inner sep=0.75pt]  [font=\large,rotate=-359.95]  {$...$};
\draw (207.73,343.21) node [anchor=north west][inner sep=0.75pt]  [font=\large,rotate=-359.95]  {$...$};
\draw (207.73,428.04) node [anchor=north west][inner sep=0.75pt]  [font=\large,rotate=-359.95]  {$...$};
\draw (586.73,286.19) node [anchor=north west][inner sep=0.75pt]  [font=\small] [align=left] {$\displaystyle e_{1,0}$};
\draw (414.92,287.19) node [anchor=north west][inner sep=0.75pt]  [font=\small] [align=left] {$\displaystyle e_{1,u-1}$};
\draw (586.73,343.12) node [anchor=north west][inner sep=0.75pt]  [font=\small] [align=left] {$\displaystyle e_{2,0}$};
\draw (585.73,426.96) node [anchor=north west][inner sep=0.75pt]  [font=\small] [align=left] {$\displaystyle e_{\bar{n} ,0}$};
\draw (414.53,425.92) node [anchor=north west][inner sep=0.75pt]  [font=\small] [align=left] {$\displaystyle e_{\bar{n} ,u-1}$};
\draw (414.92,343.08) node [anchor=north west][inner sep=0.75pt]  [font=\small] [align=left] {$\displaystyle e_{2,u-1}$};
\draw (99.87,459.08) node [anchor=north west][inner sep=0.75pt]   [align=left] {$ $};
\draw (111.72,472.03) node [anchor=north west][inner sep=0.75pt]  [font=\small,color={rgb, 255:red, 12; green, 7; blue, 250 }  ,opacity=1 ] [align=left] {RS$\displaystyle _{\rm rack}( n,k,A,u)$};
\draw (471.21,473.04) node [anchor=north west][inner sep=0.75pt]  [font=\small,color={rgb, 255:red, 11; green, 33; blue, 244 } ,opacity=1 ] [align=left] {RS$\displaystyle \left(\bar{n} ,\lceil \frac{k}{u} \rceil ,\bar{A}\right)$};
\draw (478.72,287.19) node [anchor=north west][inner sep=0.75pt]  [font=\small] [align=left] {$\displaystyle e_{1,u-2}$};
\draw (478.53,426.92) node [anchor=north west][inner sep=0.75pt]  [font=\small] [align=left] {$\displaystyle e_{\bar{n} ,u-2}$};
\draw (479.62,343.08) node [anchor=north west][inner sep=0.75pt]  [font=\small] [align=left] {$\displaystyle e_{2,u-2}$};
\draw (436.08,373.65) node [anchor=north west][inner sep=0.75pt]  [font=\large,rotate=-90.04]  {$...$};
\draw (498.08,372.65) node [anchor=north west][inner sep=0.75pt]  [font=\large,rotate=-90.04]  {$...$};
\draw (599.08,371.65) node [anchor=north west][inner sep=0.75pt]  [font=\large,rotate=-90.04]  {$...$};

\end{tikzpicture}

\caption{ The figure shows the procedure of our repair framework.}
\end{figure*}
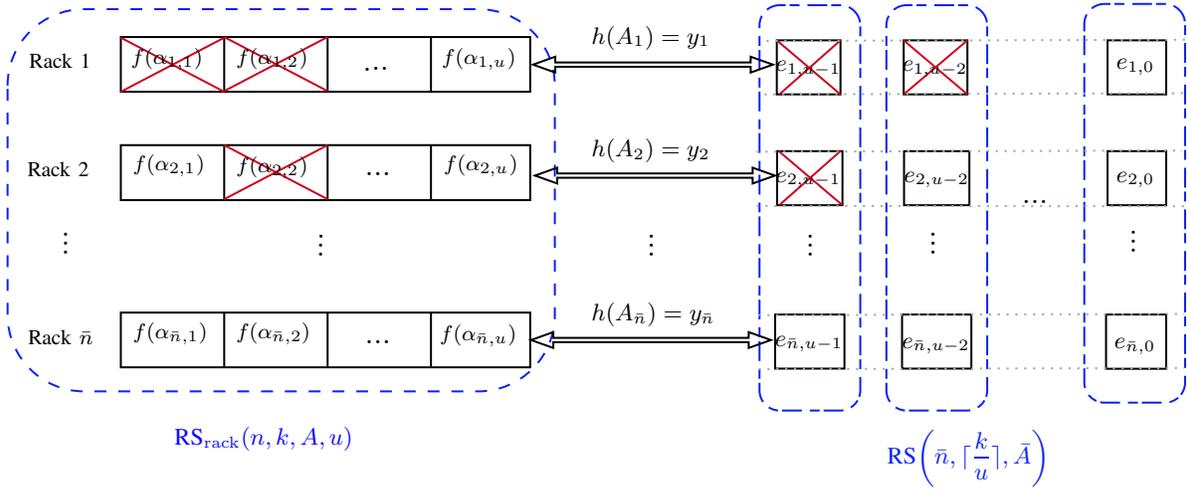
}

 Now,  we are going to introduce the repair scheme and analyze its bandwidth.
 Assume that there are {$1\leq m\leq \bar{n}-\lceil\frac{k}{u}\rceil$} racks suffering failures, indexed by the set $\{i^*_{1},\cdots,i^*_{m}\}\subseteq [\bar{n}]$.
 For $1\leq \tau\leq m$, let $\e_\tau$ denote the number of failures in rack $i^*_{\tau}$ for $1\leq \e_\tau\leq u$. {Since $ mu\leq n-u\lceil\frac{k}{u}\rceil$ and the original RS code has the minimum Hamming distance $d=n-k+1\geq n-u\lceil\frac{k}{u}\rceil+1$, this system can tolerate $mu$ failures.}
 Let $D=\{i_{1},\cdots,i_{\bar{d}}\}\subseteq [\bar{n}]\backslash\{i^*_{1},\cdots,i^*_{m}\}$ be the set of helper racks of size {$\lceil\frac{k}{u}\rceil\leq\bar{d}\leq \bar{n}-m$}.
 The following repair procedure shows the concrete repair of rack-aware RS codes in \eqref{eqn_cons}.

 {
\textbf{Repair procedure of rack-aware RS codes organized by good polynomial $h(x)$ for {multiple failures}:}
\begin{enumerate}
\item[Step 1.] For $i\in D$, the {relayer of } $i$-th helper rack computes the residue polynomials $f_{i}(x)$  through Lagrange interpolation
on  $u$ pairs $\{(\alpha_{i,j},f(\alpha_{i,j})):i\in D,j\in[u])\}$, and then gets
\begin{eqnarray*}
(e_{i,0},e_{i,1},\cdots,e_{i,u-1}).
\end{eqnarray*}

\item[Step 2.]
 Set $\epsilon=\max_{1\leq \tau\leq m}\epsilon_\tau$.
For $1\leq t\leq \epsilon$, define
\begin{equation*}\label{eqn_def_Ri}
R_t\triangleq \{i^*_{\tau}~:~\epsilon_{\tau}\geq t, 1\leq \tau\leq m\}.
\end{equation*}
Given $t\in[\epsilon]$, denote
\begin{eqnarray*}
C_{u-t}=(e_{1,u-t},e_{2,u-t},\cdots,e_{\bar{n},u-t}),\end{eqnarray*}
which is a codeword of an $[\bar{n},\lceil \frac{k}{u}\rceil]$ RS code by Theorem \ref{thm:rack_rs}.
Then, for $1\leq t\leq \epsilon$, respectively invoke the known repair schemes for homogeneous storage model to recover the $|R_{t}|$ symbols
\begin{eqnarray}\label{eq_co}
\{e_{i, u-t}: \,i\in R_t\}
\end{eqnarray}
of the RS codeword $C_{u-t}$.

\item[Step 3.] Given $i_{\tau}^*$ and $\tau\in [m]$, let the corresponding evaluation points of the $u-\epsilon_{\tau}$ surviving nodes be
$\alpha_{i_{\tau}^*,j_1},\cdots,\alpha_{i_{\tau}^*,j_{u-\epsilon_{\tau}}}$. Then, for each $i_{\tau}^*$,  from \eqref{Eqn_f_ix} one has the following  $u-\epsilon_{\tau}$ equations {
\begin{eqnarray*}\label{Eqn_f_ix2}
\sum_{j=0}^{u-\epsilon_{\tau}-1} e_{i^*_{\tau},j}\alpha_{i_{\tau}^*,j_1}^j&=&\gamma_1\nonumber,\\
&\vdots&\\
\sum_{j=0}^{u-\epsilon_{\tau}-1} e_{i^*_{\tau},j}\alpha_{i_{\tau}^*,j_{u-\epsilon_{\tau}}}^j&=&\gamma_{u-{\e_\tau}}\nonumber,
\end{eqnarray*}
where the terms on the right-hand side
\begin{equation*}
\gamma_{s}=f(\alpha_{i_{\tau}^*,j_s})-\sum_{j=u-\epsilon_{\tau}}^{u-1} e_{i^*_{\tau},j}\alpha_{i_{\tau}^*,j_s}^j, \,s\in[u-\e_\tau]
\end{equation*}}are calculated from the  coefficients $\{e_{i_{\tau}^*,u-t}:t\in [\epsilon_{\tau}]\}$  obtained in \eqref{eq_co}
and the values $\{f(\alpha_{i_{\tau}^*,j_1}),\cdots,f(\alpha_{i_{\tau}^*,j_{u-\e_{\tau}}})\}$ of the surviving nodes in rack $i_{\tau}^*$. {Note that the coefficient matrix of above equations is an invertible Vandermonde matrix (In general, the elements in $\mathcal{M}_{i^*}$ can be of the form $j,j+1,\cdots,j+u-\epsilon_{i^*}-1$ for any $j\in[0,\epsilon_{i^*}]$, then the invertibility of the coefficient matrix can be guaranteed).  Then, the relayer in rack $i_{\tau}^*$} is able to obtain the remainder coefficients $\{e_{i_{\tau}^*, u-t}: \,t\in[\e_{\tau}+1,u]\}$
and thus  the entire polynomial
\begin{eqnarray*}
f_{i^*_{\tau}}(x)=\sum_{j=0}^{u-1} e_{i^*_{\tau},j}x^j, \,\,\, \tau\in[m],
\end{eqnarray*}
which can  figure out the $\e_{\tau}$ failures
$$\{f(\alpha_{i^*_{\tau},j_{u-\e_{\tau}+1}}),\cdots,f(\alpha_{i^*_{\tau},j_{{u}}})\}.$$

\end{enumerate}

According to \eqref{Eqn_Rack_RB}, the repair bandwidth is just to count the data transmission across the rack to repair the symbols in \eqref{eq_co} in Step 2.

\begin{thm}\label{thm_general_repair} Let {${c}\in\mathcal{C}$} be a rack-aware RS  codeword in Construction \ref{cons}.
 Assume that the repair bandwidth $b_{|R_t|}$
 is capable to recover  the symbols  in  \eqref{eq_co}  from $|D|$ helper {racks} for $1\le t\le \epsilon$. Then, for the {$1\leq m\leq \lceil\frac{k}{u}\rceil$} racks $\{i_\tau^*:\tau\in[m]\}\subseteq [\bar{n}]$
 of {${c}$} containing $\{\e_\tau:\tau\in[m]\}$ failures respectively, the failures can be recovered with
  bandwidth
 \begin{equation}\label{eq: rack_bandwidth}
 b=\sum_{1\leq t\leq \e}b_{|R_t|}.
 \end{equation}
 \end{thm}}

\begin{remark}
The known repair schemes used in step 2  may be for single failure or multiple ones depending on $|R_{t}|=1$ or not, where the latter can be under the centralized model or cooperative model depending on the application scenarios.
\end{remark}

\begin{remark}
Note that the repair procedure is only related to the number of failures within the racks, i.e., $\e_1,\e_2,\ldots, \e_m$  regardless of which exact $\e_\tau$ nodes in the $i^*_\tau$-th rack experience failure. For example,  {as shown in Figure $1$,
when rack 2 contains a failure, the coefficient to be repaired remains $e_{2,u-1}$ (or $e_{2,0}$), regardless of whether any one of ${f(\alpha_{2,1}),\cdots, f(\alpha_{2,u})}$ fails.}
\end{remark}

{
\begin{remark}
Importantly, the repair procedure can be performed parallelly on repairing $\e$ codewords of an $[\bar{n},\lceil \frac{k}{u}\rceil]$ RS code from $\bar{d}$ helper racks to improve the time efficiency.  For example, as shown in Figure $1$,  these two codewords $\{e_{1,u-1},\cdots,e_{\bar{n},u-1}\} $ and  $\{e_{1,u-2},\cdots,e_{\bar{n},u-2}\}$ can be repaired in parallel, with the difference being that the former performs repair on two failures, while the latter performs repair on a single failure. In contrast, a serial repair is also feasible, whereby the previously repaired rack participates in the repair of the next one as a helper rack, with the benefit of a reduced bandwidth requirement due to the increased number of connected helper nodes.
\end{remark}
}

\begin{remark}
{ Note from  Theorem \ref{thm_general_repair} that the repair bandwidth of our  scheme   is  actually determined by the repair scheme of RS code under the homogeneous storage model. As a result, its optimality depends on the repair bandwidth of the corresponding homogeneous RS code.

For the single-rack failure, if the homogeneous RS code is optimal for repairing a single failed node, it is easy to verify from \eqref{eq: rack_bandwidth} that our repair scheme for rack-aware RS code can achieve the cut-set bound in \cite{CB19} with equality when $u|k$.

For multiple-rack failures, to the best of our knowledge, currently, there are no known bounds in the literature for the repair bandwidth concerning multiple-rack failures. Therefore, the optimality for the multiple-rack failures remains an open problem. }
\end{remark}

\begin{remark}
 {Another related work on repairing Tamo-Barg codes based on good polynomials
 is introduced in \cite{SL22}, which reduces the bandwidth within repair sets. However, it is a locally repairable code without MDS property.}
\end{remark}

\section{Some Explicit Rack-Aware Reed-Solomon Codes by means of good polynomials}\label{sec: construction}
{ In this section, we mainly focus on the case that multiple failures take place within one rack since most of the known results on repairing RS codes under the homogeneous model are for a single failure. Therefore, firstly, we give the explicit repair procedure for rack-aware RS codes organized by good polynomials containing failures within a single rack. Using this repair framework, we provide a multiple-node repair scheme of single-rack failure for existing constructions that can only tolerate single-node failure before. Finally, as an illustration, we use three classes of good polynomials to deduce several rack-aware RS codes supporting single-rack and multiple-rack failures respectively.
}
\subsection{The Repair of Single-Rack Failure}\label{sec: cons_1}
We first simplify the repair procedure to the case of single-rack failure.

 \textbf{Repair procedure of rack-aware RS codes organized by good polynomial $h(x)$ for failures within a single-rack}:

 {
Since Steps $1$ and $3$ are the same as those in the general procedure,  herein we only refine Step $2$.

\begin{enumerate}

\item[Step 2.]
Let $\e$ be the number of failures in $i^*$-th rack.
In this case,
$$R_{1}=R_{2}=\cdots=R_{\e}=\{i^*\}.$$
Given $t\in[\e]$, denote
\begin{eqnarray*}
C_{u-t}=(e_{1,u-t},e_{2,u-t},\cdots,e_{\bar{n},u-t})
\end{eqnarray*}
as an $[\bar{n},\lceil \frac{k}{u}\rceil]$ RS codeword under the homogeneous storage model. Then, only $1$ coordinate $e_{i^*,u-t}$ of  $C_{u-t}$ needs to be repaired.
By employing the known single-node repair schemes under the homogeneous storage model (e.g., \cite{BBD+22,TY17,YB16,DDK+17,GW17,LWJ19}) to  the $\e$ codewords $C_{u-t}$,  obtain all the desired coefficients
\begin{eqnarray*}\label{eq: co}
\{e_{i^*,u-1},\cdots,e_{i^*,u-\e}\}.
\end{eqnarray*}

\end{enumerate}
}

For the repair scheme, we have the following conclusion, which is a direct corollary of Theorem \ref{thm_general_repair} and the proof
is omitted.

\begin{corollary}\label{coro:rack_b}
Let {$c\in\mathcal{C}$} be a rack-aware RS  codeword in Construction \ref{cons}.
For any $i^*\in [\bar{n}]$ and $1\leq \e\leq u$, the $\e$ failed nodes of the $i^*$-th rack in $c$ can be recovered by the above repair procedure, whose repair bandwidth can be given as
\begin{eqnarray*}
b= \e b',
\end{eqnarray*}
where $b'$ denotes the repair bandwidth of a single failure with $|D|$ helper nodes in Step 2.
\end{corollary}

Then we demonstrate our repair procedure with an illustrative example.

{
\begin{ex}\label{ex_Jin}
Let $\mathcal{C}={\rm RS}(n=2^4,k=7,A=\f_{2^4},)$ be an RS code over $\f_{2^4}$. Define a good polynomial
$$h(x)=\tr_{\f_{2^4}/\f_{2^2}}(x)=x+x^4,$$
with kernel  $U=\{0,1,\gamma^5,\gamma^{10}\}$, where $\gamma$ is a root of primitive polynomial $x^4+x+1$ of $\f_{2^4}$.   Obviously, $A=\f_{2^4}=\cup_{i=1}^4A_{i}$ can be partitioned into $\bar{n}=4$ distinct sets $A_1=U,A_2=\gamma+U,A_3=\gamma^{6}+U,A_4=\gamma^{3}+U$
with $h(A_i)=0,1,\gamma^5,\gamma^{10}$
respectively. Then, according to \eqref{eqn_cons}, the $n=2^4$ nodes can be organized into $\bar{n}=4$ racks based on $h(x)$, each containing $u=4$ nodes, i.e.,
 \begin{eqnarray}\label{Eqn_Example_1}
\begin{pmatrix}
  f_1(0)&f_1(1)&f_{1}(\gamma^5)&f_1(\gamma^{10})\\
  f_2(\gamma)&f_2(\gamma^2)&f_2(\gamma^4)&f_2(\gamma^8)\\
   f_3(\gamma^6)&f_3(\gamma^7)&f_3(\gamma^9)&f_3(\gamma^{13})\\
   f_4(\gamma^3)&f_4(\gamma^{11})&f_4(\gamma^{12})&f_4(\gamma^{14})\\
     \end{pmatrix},
  \end{eqnarray}
where $f_i(x)\equiv f(x)\pmod{h(x)-y_i}$ with $\deg{(f_i)}<4$ for $i\in [4]$,

Note that  $i$-th row in \eqref{Eqn_Example_1} corresponds to the $i$-th rack. Without loss of generality, assume that $3$ nodes $f(1),f(\gamma^5),f(\gamma^{10})$ in $1$-th rack fail and $D=\{2,3,4\}$ is the set of helper racks, {which is shown in Fig. $2$}. The repair procedure can be performed as follows.

\begin{figure*}[ht]\label{fig_ex}
\centering

\tikzset{every picture/.style={line width=0.75pt}} 

\begin{tikzpicture}[x=0.75pt,y=0.75pt,yscale=-1,xscale=1]

\draw   (63.11,122.57) -- (269.58,122.57) -- (269.58,150.19) -- (63.11,150.19) -- cycle ;
\draw   (580.76,124.23) -- (610.33,124.23) -- (610.33,151.52) -- (580.76,151.52) -- cycle ;
\draw   (580.76,172.16) -- (610.33,172.16) -- (610.33,199.45) -- (580.76,199.45) -- cycle ;
\draw   (580,272.63) -- (610.33,272.63) -- (610.33,300) -- (580,300) -- cycle ;
\draw   (387.95,124.23) -- (420.07,124.23) -- (420.07,151.52) -- (387.95,151.52) -- cycle ;
\draw   (388.12,172.17) -- (421.07,172.17) -- (421.07,199.45) -- (388.12,199.45) -- cycle ;
\draw   (387.12,272.63) -- (422.07,272.63) -- (422.07,299.92) -- (387.12,299.92) -- cycle ;
\draw  [color={rgb, 255:red, 11; green, 33; blue, 244 } ,draw opacity=1 ][dash pattern={on 3.75pt off 3pt on 7.5pt off 1.5pt}] (569.03,116.03) .. controls (569.03,110.41) and (573.58,105.86) .. (579.2,105.86) -- (609.7,105.86) .. controls (615.32,105.86) and (619.87,110.41) .. (619.87,116.03) -- (619.87,310.33) .. controls (619.87,315.95) and (615.32,320.5) .. (609.7,320.5) -- (579.2,320.5) .. controls (573.58,320.5) and (569.03,315.95) .. (569.03,310.33) -- cycle ;
\draw  [color={rgb, 255:red, 11; green, 33; blue, 244 } ,draw opacity=1 ][dash pattern={on 3.75pt off 3pt on 7.5pt off 1.5pt}] (378.33,115.64) .. controls (378.33,110.02) and (382.88,105.47) .. (388.49,105.47) -- (419,105.47) .. controls (424.61,105.47) and (429.17,110.02) .. (429.17,115.64) -- (429.17,310.33) .. controls (429.17,315.95) and (424.61,320.5) .. (419,320.5) -- (388.49,320.5) .. controls (382.88,320.5) and (378.33,315.95) .. (378.33,310.33) -- cycle ;
\draw    (115.25,122.57) -- (115.25,149.85) ;
\draw    (167.39,122.57) -- (167.39,149.85) ;
\draw    (219.54,122.57) -- (219.54,149.85) ;
\draw   (63.11,172.14) -- (269.58,172.14) -- (269.58,199.88) -- (63.11,199.88) -- cycle ;
\draw    (115.25,172.14) -- (115.25,199.42) ;
\draw    (167.39,172.14) -- (167.39,199.42) ;
\draw    (219.54,172.14) -- (219.54,199.42) ;
\draw   (63.11,272.65) -- (269.58,272.65) -- (269.58,300.27) -- (63.11,300.27) -- cycle ;
\draw    (115.25,272.65) -- (115.25,299.94) ;
\draw    (167.39,272.65) -- (167.39,299.94) ;
\draw    (219.54,272.65) -- (219.54,299.94) ;
\draw [color={rgb, 255:red, 208; green, 2; blue, 27 }  ,draw opacity=1 ]   (115.25,122.57) -- (167.39,149.85) ;
\draw [color={rgb, 255:red, 208; green, 2; blue, 27 }  ,draw opacity=1 ]   (115.25,149.85) -- (167.39,122.57) ;
\draw [color={rgb, 255:red, 208; green, 2; blue, 27 }  ,draw opacity=1 ]   (389.23,151.52) -- (420.07,124.23) ;
\draw [color={rgb, 255:red, 208; green, 2; blue, 27 }  ,draw opacity=1 ]   (387.95,124.23) -- (420.07,151.52) ;
\draw   (448.95,124.23) -- (481.07,124.23) -- (481.07,151.52) -- (448.95,151.52) -- cycle ;
\draw   (449.12,172.17) -- (482.07,172.17) -- (482.07,199.45) -- (449.12,199.45) -- cycle ;
\draw   (448.12,272.63) -- (483.07,272.63) -- (483.07,299.92) -- (448.12,299.92) -- cycle ;
\draw  [color={rgb, 255:red, 11; green, 33; blue, 244 } ,draw opacity=1 ][dash pattern={on 3.75pt off 3pt on 7.5pt off 1.5pt}] (440.22,116.67) .. controls (440.22,111.05) and (444.77,106.5) .. (450.39,106.5) -- (480.89,106.5) .. controls (486.5,106.5) and (491.06,111.05) .. (491.06,116.67) -- (491.06,310.33) .. controls (491.06,315.95) and (486.5,320.5) .. (480.89,320.5) -- (450.39,320.5) .. controls (444.77,320.5) and (440.22,315.95) .. (440.22,310.33) -- cycle ;
\draw [color={rgb, 255:red, 208; green, 2; blue, 27 }  ,draw opacity=1 ]   (448.95,151.52) -- (481.07,124.23) ;
\draw [color={rgb, 255:red, 208; green, 2; blue, 27 }  ,draw opacity=1 ]   (481.07,151.52) -- (448.95,124.23) ;
\draw  [color={rgb, 255:red, 11; green, 33; blue, 244 }  ,draw opacity=1 ][dash pattern={on 4.5pt off 4.5pt}] (4,150.7) .. controls (4,127.39) and (22.89,108.5) .. (46.2,108.5) -- (237.8,108.5) .. controls (261.11,108.5) and (280,127.39) .. (280,150.7) -- (280,277.3) .. controls (280,300.61) and (261.11,319.5) .. (237.8,319.5) -- (46.2,319.5) .. controls (22.89,319.5) and (4,300.61) .. (4,277.3) -- cycle ;
\draw [color={rgb, 255:red, 208; green, 2; blue, 27 }  ,draw opacity=1 ]   (167.39,149.85) -- (219.54,122.57) ;
\draw [color={rgb, 255:red, 208; green, 2; blue, 27 }  ,draw opacity=1 ]   (167.39,122.57) -- (219.54,149.85) ;
\draw [color={rgb, 255:red, 153; green, 178; blue, 173 }  ,draw opacity=1 ][line width=0.75]  [dash pattern={on 0.84pt off 2.51pt}]  (379,124.5) -- (624.4,122.92) ;
\draw [color={rgb, 255:red, 153; green, 178; blue, 173 }  ,draw opacity=1 ][line width=0.75]  [dash pattern={on 0.84pt off 2.51pt}]  (380,151.5) -- (623.4,150.92) ;
\draw [color={rgb, 255:red, 153; green, 178; blue, 173 }  ,draw opacity=1 ][line width=0.75]  [dash pattern={on 0.84pt off 2.51pt}]  (379,172.5) -- (617.4,172.85) ;
\draw [color={rgb, 255:red, 153; green, 178; blue, 173 }  ,draw opacity=1 ][line width=0.75]  [dash pattern={on 0.84pt off 2.51pt}]  (379,199.5) -- (618,198) ;
\draw [color={rgb, 255:red, 153; green, 178; blue, 173 }  ,draw opacity=1 ][line width=0.75]  [dash pattern={on 0.84pt off 2.51pt}]  (379,300.5) -- (624.32,300.68) ;
\draw [color={rgb, 255:red, 153; green, 178; blue, 173 }  ,draw opacity=1 ][line width=0.75]  [dash pattern={on 0.84pt off 2.51pt}]  (378,272.5) -- (623.32,272.31) ;
\draw   (516.76,124.23) -- (546.33,124.23) -- (546.33,151.52) -- (516.76,151.52) -- cycle ;
\draw   (516.76,172.16) -- (546.33,172.16) -- (546.33,199.45) -- (516.76,199.45) -- cycle ;
\draw   (516,272.63) -- (546.33,272.63) -- (546.33,300) -- (516,300) -- cycle ;
\draw  [color={rgb, 255:red, 11; green, 33; blue, 244 }  ,draw opacity=1 ][dash pattern={on 3.75pt off 3pt on 7.5pt off 1.5pt}] (505.03,116.03) .. controls (505.03,110.41) and (509.58,105.86) .. (515.2,105.86) -- (545.7,105.86) .. controls (551.32,105.86) and (555.87,110.41) .. (555.87,116.03) -- (555.87,310.33) .. controls (555.87,315.95) and (551.32,320.5) .. (545.7,320.5) -- (515.2,320.5) .. controls (509.58,320.5) and (505.03,315.95) .. (505.03,310.33) -- cycle ;
\draw   (63.11,222.14) -- (269.58,222.14) -- (269.58,249.88) -- (63.11,249.88) -- cycle ;
\draw    (115.25,222.14) -- (115.25,249.42) ;
\draw    (167.39,222.14) -- (167.39,249.42) ;
\draw    (219.54,222.14) -- (219.54,249.42) ;
\draw   (580.76,221.16) -- (610.33,221.16) -- (610.33,248.45) -- (580.76,248.45) -- cycle ;
\draw   (388.12,221.17) -- (421.07,221.17) -- (421.07,248.45) -- (388.12,248.45) -- cycle ;
\draw   (449.12,221.17) -- (482.07,221.17) -- (482.07,248.45) -- (449.12,248.45) -- cycle ;
\draw [color={rgb, 255:red, 153; green, 178; blue, 173 }  ,draw opacity=1 ][line width=0.75]  [dash pattern={on 0.84pt off 2.51pt}]  (379,221.5) -- (617.4,221.85) ;
\draw [color={rgb, 255:red, 153; green, 178; blue, 173 }  ,draw opacity=1 ][line width=0.75]  [dash pattern={on 0.84pt off 2.51pt}]  (381,248.5) -- (620,247) ;
\draw   (516.76,221.16) -- (546.33,221.16) -- (546.33,248.45) -- (516.76,248.45) -- cycle ;
\draw [color={rgb, 255:red, 208; green, 2; blue, 27 }  ,draw opacity=1 ]   (217.44,122.9) -- (269.58,150.19) ;
\draw [color={rgb, 255:red, 208; green, 2; blue, 27 }  ,draw opacity=1 ]   (217.44,150.19) -- (269.58,122.9) ;
\draw [color={rgb, 255:red, 208; green, 2; blue, 27 }  ,draw opacity=1 ]   (516.76,151.52) -- (546.33,124.23) ;
\draw [color={rgb, 255:red, 208; green, 2; blue, 27 }  ,draw opacity=1 ]   (546.33,151.52) -- (516.76,124.23) ;
\draw   (270,137.75) -- (280.86,132.5) -- (280.86,136.32) -- (376.14,136.32) -- (376.14,132.5) -- (387,137.75) -- (376.14,143) -- (376.14,139.18) -- (280.86,139.18) -- (280.86,143) -- cycle ;
\draw   (270,185.75) -- (280.86,180.5) -- (280.86,184.32) -- (376.14,184.32) -- (376.14,180.5) -- (387,185.75) -- (376.14,191) -- (376.14,187.18) -- (280.86,187.18) -- (280.86,191) -- cycle ;
\draw   (270,234.75) -- (280.86,229.5) -- (280.86,233.32) -- (376.14,233.32) -- (376.14,229.5) -- (387,234.75) -- (376.14,240) -- (376.14,236.18) -- (280.86,236.18) -- (280.86,240) -- cycle ;
\draw   (270,287.75) -- (280.86,282.5) -- (280.86,286.32) -- (376.14,286.32) -- (376.14,282.5) -- (387,287.75) -- (376.14,293) -- (376.14,289.18) -- (280.86,289.18) -- (280.86,293) -- cycle ;

\draw (15.54,129.38) node [anchor=north west][inner sep=0.75pt]  [font=\large] [align=left] {{\footnotesize Rack 1}};
\draw (14.54,179.22) node [anchor=north west][inner sep=0.75pt]  [font=\large] [align=left] {{\footnotesize Rack 2}};
\draw (13.99,280.1) node [anchor=north west][inner sep=0.75pt]  [font=\large] [align=left] {{\footnotesize Rack 4}};
\draw (299.07,115.41) node [anchor=north west][inner sep=0.75pt]  [font=\small] [align=left] {$\displaystyle h( A_{1}) =0$};
\draw (76.86,128.43) node [anchor=north west][inner sep=0.75pt]  [font=\footnotesize] [align=left] {$\displaystyle f( 0)$};
\draw (128.78,129.43) node [anchor=north west][inner sep=0.75pt]  [font=\footnotesize] [align=left] {$\displaystyle f( 1)$};
\draw (228.61,128.43) node [anchor=north west][inner sep=0.75pt]  [font=\footnotesize] [align=left] {$\displaystyle f(\gamma^{10})$};
\draw (300.07,163.3) node [anchor=north west][inner sep=0.75pt]  [font=\small] [align=left] {$\displaystyle h( A_{2}) =1$};
\draw (298.06,261.14) node [anchor=north west][inner sep=0.75pt]  [font=\small] [align=left] {$\displaystyle h( A_{4}) =\gamma ^{10}$};
\draw (75.86,177.33) node [anchor=north west][inner sep=0.75pt]  [font=\footnotesize] [align=left] {$\displaystyle f( \gamma ^{})$};
\draw (73.86,277.16) node [anchor=north west][inner sep=0.75pt]  [font=\footnotesize] [align=left] {$\displaystyle f( \gamma ^{3})$};
\draw (124.78,176.33) node [anchor=north west][inner sep=0.75pt]  [font=\footnotesize] [align=left] {$\displaystyle f( \gamma ^{2})$};
\draw (228.61,176.33) node [anchor=north west][inner sep=0.75pt]  [font=\footnotesize] [align=left] {$\displaystyle f( \gamma ^{8})$};
\draw (123.78,277.16) node [anchor=north west][inner sep=0.75pt]  [font=\footnotesize] [align=left] {$\displaystyle f(\gamma ^{11})$};
\draw (227.61,277.16) node [anchor=north west][inner sep=0.75pt]  [font=\footnotesize] [align=left] {$\displaystyle f(\gamma ^{14})$};
\draw (583.73,132.19) node [anchor=north west][inner sep=0.75pt]  [font=\small] [align=left] {$\displaystyle e_{1,0}$};
\draw (391.92,133.19) node [anchor=north west][inner sep=0.75pt]  [font=\small] [align=left] {$\displaystyle e_{1,3}$};
\draw (583.73,181.12) node [anchor=north west][inner sep=0.75pt]  [font=\small] [align=left] {$\displaystyle e_{2,0}$};
\draw (582.73,282.96) node [anchor=north west][inner sep=0.75pt]  [font=\small] [align=left] {$\displaystyle e_{4,0}$};
\draw (392.53,281.92) node [anchor=north west][inner sep=0.75pt]  [font=\small] [align=left] {$\displaystyle e_{4,3}$};
\draw (392.92,181.08) node [anchor=north west][inner sep=0.75pt]  [font=\small] [align=left] {$\displaystyle e_{2,3}$};
\draw (76.87,305.08) node [anchor=north west][inner sep=0.75pt]   [align=left] {$ $};
\draw (81.72,329.03) node [anchor=north west][inner sep=0.75pt]  [font=\small,color={rgb, 255:red, 12; green, 7; blue, 250 }  ,opacity=1 ] [align=left] {RS$\displaystyle _{\rm rack}\left( 2^{4} ,7,\mathbb{F}_{2^{4}} ,4\right)$};
\draw (464.21,332.04) node [anchor=north west][inner sep=0.75pt]  [font=\small,color={rgb, 255:red, 11; green, 33; blue, 244 },opacity=1 ] [align=left] {RS$\displaystyle ( 4,2,\mathbb{F}_{2^{2}})$};
\draw (452.92,133.19) node [anchor=north west][inner sep=0.75pt]  [font=\small] [align=left] {$\displaystyle e_{1,2}$};
\draw (453.53,282.92) node [anchor=north west][inner sep=0.75pt]  [font=\small] [align=left] {$\displaystyle e_{4,2}$};
\draw (453.92,181.08) node [anchor=north west][inner sep=0.75pt]  [font=\small] [align=left] {$\displaystyle e_{2,2}$};
\draw (519.73,132.19) node [anchor=north west][inner sep=0.75pt]  [font=\small] [align=left] {$\displaystyle e_{1,1}$};
\draw (519.73,181.12) node [anchor=north west][inner sep=0.75pt]  [font=\small] [align=left] {$\displaystyle e_{2,1}$};
\draw (518.73,282.96) node [anchor=north west][inner sep=0.75pt]  [font=\small] [align=left] {$\displaystyle e_{4,1}$};
\draw (73.86,226.33) node [anchor=north west][inner sep=0.75pt]  [font=\footnotesize] [align=left] {$\displaystyle f(\gamma ^{6})$};
\draw (125.78,226.33) node [anchor=north west][inner sep=0.75pt]  [font=\footnotesize] [align=left] {$\displaystyle f(\gamma ^{7})$};
\draw (228.61,228.33) node [anchor=north west][inner sep=0.75pt]  [font=\footnotesize] [align=left] {$\displaystyle f( \gamma ^{13})$};
\draw (297.07,212.3) node [anchor=north west][inner sep=0.75pt]  [font=\small] [align=left] {$\displaystyle h( A_{3}) =\gamma ^{5}$};
\draw (583.73,230.12) node [anchor=north west][inner sep=0.75pt]  [font=\small] [align=left] {$\displaystyle e_{3,0}$};
\draw (392.92,230.08) node [anchor=north west][inner sep=0.75pt]  [font=\small] [align=left] {$\displaystyle e_{3,3}$};
\draw (453.92,230.08) node [anchor=north west][inner sep=0.75pt]  [font=\small] [align=left] {$\displaystyle e_{3,2}$};
\draw (519.73,230.12) node [anchor=north west][inner sep=0.75pt]  [font=\small] [align=left] {$\displaystyle e_{3,1}$};
\draw (175.78,128.33) node [anchor=north west][inner sep=0.75pt]  [font=\footnotesize] [align=left] {$\displaystyle f( \gamma ^{5})$};
\draw (176.78,176.33) node [anchor=north west][inner sep=0.75pt]  [font=\footnotesize] [align=left] {$\displaystyle f( \gamma ^{4})$};
\draw (176.78,226.33) node [anchor=north west][inner sep=0.75pt]  [font=\footnotesize] [align=left] {$\displaystyle f(\gamma ^{9})$};
\draw (175.78,276.33) node [anchor=north west][inner sep=0.75pt]  [font=\footnotesize] [align=left] {$\displaystyle f( \gamma ^{12})$};
\draw (14,229.1) node [anchor=north west][inner sep=0.75pt]  [font=\large] [align=left] {{\footnotesize Rack 3}};

\end{tikzpicture}

\caption{ An example of the single-rack failure.}
\end{figure*}
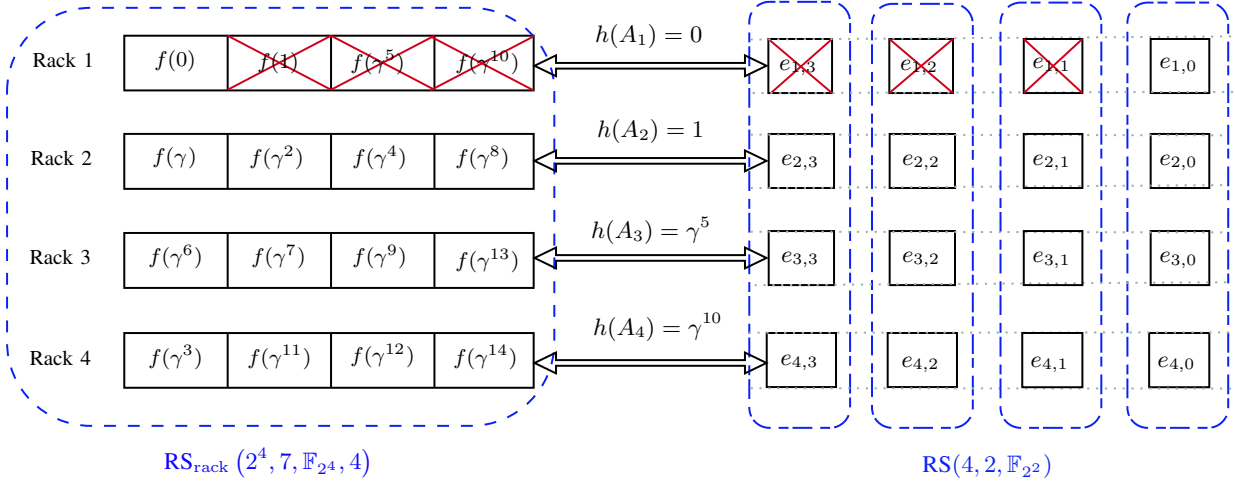

\begin{enumerate}
\item[Step 1.] The relayer in helper rack $i\in \{2,3,4\}$ computes the residue polynomials $f_{i}(x)$ and then obtains its coefficients
$$(e_{i,0},e_{i,1},e_{i,2},e_{i,3}).$$
\item[Step 2.] For $t\in [3]$, {
\begin{eqnarray*}
C_{4-t}=(e_{1,4-t},e_{2,4-t},e_{3,4-t},e_{4,4-t})
\end{eqnarray*}
forms a codeword of an $[\bar{n}=4,{\lceil {k\over u}\rceil}=2]$ RS code over $\f_{2^4}$ with evaluation points  $\bar{A}=\f_{2^2}$ according to Theorem~\ref{thm:rack_rs}. Since {$\bar{n}-{\lceil {k\over u}\rceil}\geq 2^{2-1}$}, {we employ the  single-node repair scheme in \cite{GW17}.} Let the dual codewords defined by polynomials
\begin{equation*}
\begin{split}
\bar{g}_{m,w}(x)&=\frac{\eta_{m}\tr_{\f_{2^{2}}/\f_{2}}(\beta_{w}(x-y_{1}))}{x-y_{1}}\\
&=\frac{\eta_{m}\tr_{\f_{2^{2}}/\f_{2}}(\beta_{w}x)}{x},\,\,\,m\in[2],\,w\in[2],
\end{split}
\end{equation*}
where $\{\eta_{m}:m\in[2]\}=\{1,\gamma\}$ is a basis of $\f_{2^4}$ over $\f_{2^2}$, and $\{\beta_{w}:w\in[{2}]\}=\{1,\gamma^{5}\}$ is a basis of $\f_{2^2}$ over $\f_{2}$.
Then, for any $t\in[3]$, the repair center in rack $1$ can download data given in Table \ref{tb: ex}
 from the relayers in helper racks to obtain $\{\tr_{\f_{2^{4}}/\f_2}(\eta_{m}\beta_{w}e_{1,4-t}):m\in[2];w\in[2];t\in[3]\}
$. Since $\{\eta_{m}\beta_{w}:m\in[2];w\in[2]\}=\{1,\gamma^{},\gamma^{5},\gamma^{6}\}$ forms a basis of $\f_{2^4}$ over $\f_{2}$,  one can recover $e_{1,3},e_{1,2},e_{1,1}$ by
means of~(\ref{eq:dual basis}).}
\begin{table}[ht]\label{tb: ex}
  \centering

    \caption{Download Symbols}
    \begin{tabular}{|c|c|}
      \hline
      Helper Rack & Download \\
      \hline
     Rack $2$ &$\tr_{\f_{2^{4}}/\f_2}(e_{2,4-t}),\tr_{\f_{2^{4}}/\f_2}(\gamma e_{2,4-t})$ \\
      \hline
      Rack $3$ & $\tr_{\f_{2^{4}}/\f_2}(\frac{e_{3,4-t}}{\gamma^{5}}),\tr_{\f_{2^{4}}/\f_2}(\frac{e_{3,4-t}}{\gamma^{4}})$ \\
      \hline
       Rack $4$ & $\tr_{\f_{2^{4}}/\f_2}(\frac{e_{4,4-t}}{\gamma^{10}}),\tr_{\f_{2^{4}}/\f_2}(\frac{e_{4,4-t}}{\gamma^9}) $ \\
             \hline
    \end{tabular}
  \end{table}
\item[Step 3.] Based on $\{e_{1,4-t}:t\in[3]\}$ obtained in Step $2$ {and the surviving node $f(0)$, the repair center is able to get
the polynomial
$$f_{1}(x)=e_{1,3}x^3+e_{1,2}x^2+e_{1,1}x+e_{1,0}$$
with $e_{1,0}=f(0)$, and thus repairs the $3$ failures
$$\{f(1),f(\gamma^5),f(\gamma^{10})\}.$$
}
\end{enumerate}

By Corollary~\ref{coro:rack_b}, the repair bandwidth \begin{eqnarray*}
b=\e b'=18
\end{eqnarray*}
bits.
\end{ex}

}

\subsection{Interpretation of Known Constructions}

In \cite{JLX19,CB19}, {the rack-based RS codes were designed to handle a single failure.} Actually,  these two constructions also satisfy the condition in Theorem~\ref{thm:rack_rs}, which indicates that they are applicable to our repair framework. Therefore, we will clarify the explicit good polynomial that organized nodes into racks and generalize the present schemes of \cite{JLX19} and \cite{CB19} to repair any number of failures within a single rack.

{ \textbf{$\bullet$ Single-Rack Repair for Jin {\it et al.}'s Construction} \cite{JLX19}

 In \cite{JLX19}, Jin {\it et al.} presented three classes of RS codes for the rack-aware storage model employing the good polynomials proposed in \cite{TB14}. For simplicity, we take the construction in [11, Theorem 5.3] as an example to demonstrate that they meet our repair framework.

Consider an  RS code $\mathcal{C}={\rm{RS}}(n,k,A)$ over $
\fqt$, where $n=q^t$ and $2|t$.

\begin{enumerate}
\item[1)] Let $$h(x)=\tr_{\fqt/ \f_{q^2}}(x)-\tr_{\fqt/ \f_{q^2}}(\beta),$$ then $h(A_{i})=h(a_{i}+W)=h(a_{i})$, where $W=ker(\tr_{\fqt/ \f_{q^2}}(x))$, $A_{i}=\{a_{i}+W:i\in[q^2]\}$ are distinct cosets of $W$ satisfying $A=\fqt=\cup_{i=1}^{q^2} A_i$, and $\beta$ is an element in a subset $A_{i}$. Then, the $u=|W|=q^{t-2}$ nodes
\begin{eqnarray*}
\{f(a_{i}+w):w\in W\}
\end{eqnarray*} are stored in $i$-th rack, $i\in[q^2]$.
\item[2)] Suppose that $\e$ failures occur in the $i^*$-th rack.
According to the repair procedure in Section \ref{sec: cons_1}, the repair problem turns out to repair a single failure of an $[\bar{n},\lceil\frac{k}{u}\rceil]$ RS code defined by $\bar{A}=\{y_{i}:i\in[\bar{n}]\}=\{\tr_{\fqt/ \f_{q^2}}(a_{i})-\tr_{\fqt/ \f_{q^2}}(\beta):i\in[\bar{n}]\}$ under the homogeneous storage model. Suppose that $\bar{n}-\lceil\frac{k}{u}\rceil\geq q^{\bar{s}}$, $1\leq \bar{s}<t$. It then satisfies the prerequisite of the repair scheme given in \cite{DDK+17}. Hence, we
let
\begin{eqnarray*}
\bar{g}_{l}(x)=\frac{{\rm L_{\bar{U}}}(\beta_{l}(x-y_{i^*}))}{x-y_{i^*}},\,\,\,l\in[t],
\end{eqnarray*}where $\bar{U}$ is a subspace of $\fqt$ of dimension $\bar{s}$.
Therefore, for $j\in[\e]$, rack $i^*$ downloads
{\begin{equation*}
\begin{split}
b'=|\{\tr(\gamma_{i,w}e_{i,{u-j}}):i\in[\bar{n}]\backslash\{i^*\},&w\in[b_{i}]\}|\\
&\leq(\bar{n}-1)(t-\bar{s})
\end{split}
\end{equation*}
 sub-symbols from the helper racks to repair the $\e$ coefficients
\begin{eqnarray}\label{Eq_Co_Jin}
\{e_{i^*,u-\e},\cdots,e_{i^*,u-1}\}
\end{eqnarray}
of $f_{i^*}(x)$, where $\{\gamma_{i,w}:w\in[b_{i}]\}$ forms a basis of ${\rm Span}_{\fq }(\frac{{\rm L_{\bar{U}}}(\beta_{l}(y_{i}-y_{i^*}))}{y_{i}-y_{i^*}}:l\in[t])$ over $\fq$ and $b_{i}={\rm Rank}_{\fq}\{\frac{{\rm L_{\bar{U}}}(\beta_{l}(y_{i}-y_{i^*}))}{y_{i}-y_{i^*}}:l\in[t]\}\leq t-\bar{s}$. }
Then the failed symbols can be repaired by the coefficients obtained from (\ref{Eq_Co_Jin})  and the surviving nodes within the rack.
\end{enumerate}

\begin{corollary}
The rack-aware RS codes given in \cite{JLX19} with $1\leq \e\leq u$ failures located on the same rack can be repaired with bandwidth at most
$\e(\bar{n}-1)(t-\bar{s})$ sub-symbols.
\end{corollary}

When $\e=1$, {our repair scheme has the same bandwidth $b'$ as Jin {\it et al.}'s.}
It is worth noting that Jin {\it et al.}  had applied good polynomials to design rack-aware RS codes, but they did not explicitly explore the potential repair property of the residue polynomials. As a result, the original scheme of \cite{JLX19} only supports a single failure, whereas our repair framework is capable of handling multiple failures.


\textbf{$\bullet$  Single-Rack Repair for Chen and Barg's Construction}~\cite{CB19}

Let $k=u\bar{k}+v,$ $0\leq v\leq u-1$ and $q$ be a power of a prime with $u|(q-1)$. Let $p_{i}$ be distinct primes such that $p_{i}\equiv1$ mod $\bar{s}=\bar{d}-k'+1$ and $p_{i}>u$ for $i\in[\bar{n}]$, where $k'=\lceil\frac{k}{u}\rceil$. We take the smallest $\bar{n}$ primes. For any $i\in[\bar{n}]$, let $\lambda_{i}$ be an element of degree $p_{i}$ over $\fq$, namely, $[\fq(\lambda_{i}),\fq]=p_{i}$. Define the finite field $F_{i}:=\fq(\lambda_{j}:j\in[\bar{n}]\backslash\{i\})$ and $\mathbb{F}:=\fq(\lambda_{1},\cdots,\lambda_{\bar{n}})$,  i.e., $\mathbb{F}=F_{i}(\lambda_{i})$.
Let $\mathbb{K}$ be an extension of $\mathbb{F}$ with degree $\bar{s}$. It is known that $\mathbb{K}$ is also an extension of $\fq$ with degree $t=[\mathbb{K}:\fq]=\bar{s}\Pi_{i=1}^{\bar{n}}p_{i}$.

Consider an  RS code $\mathcal{C}={\rm{RS}}(n,k,A)$ over $
\fqt$ with
$A=\cup_{i=1}^{\bar{n}}\{\lambda_{ij}:j\in[u]\},$
where $\lambda_{ij}=\lambda_{i}\lambda^{j-1}$ and $\lambda\in\fq$ is an element of order $u$.
Since $\lambda_{i}$ is a generator of $\mathbb{F}_{q^{p_{i}}}$ over $\fq$, and $p_{i}$ is a prime with $p_{i}>u$, the element $\lambda_{i}^u$ satisfies $\fq(\lambda_{i}^u)=\mathbb{F}_{q^{p_{i}}}$ as well .

\begin{enumerate}
\item[1)] Let $$h(x)=x^u,$$ then $h(A_{i})=\lambda_{i}^u,$ where $A_{i}=\{\lambda_{ij}:j\in[u]\}$ and $i\in[\bar{n}]$. Then, place the $u$ nodes corresponding to $A_{i}$ into $i$-th rack, i.e.
\begin{eqnarray*}
(f(\lambda_{i1}),f(\lambda_{i2}),\cdots,f(\lambda_{iu}))
\end{eqnarray*} are stored in $i$-th rack, $i\in[\bar{n}]$.

\item[2)]  Suppose that $\e$ failures occur in the $i^*$-th rack.  $D\subset [\bar{n}]\backslash\{i^*\}$ represents the set of helper racks of size $\bar{d}$. According to the repair procedure given in Section \ref{sec: cons_1}, the repair problem is transformed to {repairing} a single failure of an $[\bar{n},k']$ Reed-Solomon code defined by $\bar{A}=\{\lambda_{i}^u:i\in[\bar{n}]\}$ under the homogeneous storage model, which has optimal repair property given in \cite{TY19}.
Let
\begin{eqnarray*}
\bar{g}_{l}(x)=\bar{g}_{w,\nu}(x)=\gamma_{w}x^{\nu}\Pi_{i\in[\bar{n}]\backslash\{\{i^*\}\cup D\}}(x-\lambda_{i}^u),
\end{eqnarray*}
where $\nu\in[0,\bar{s}-1]$ and $\{\gamma_{w}:w\in[p_{i^*}]\}$ are elements of $\mathbb{K}$ such that $\{\gamma_{w}\lambda_{i^*}^{u\nu}:w\in[p_{i^*}],\nu\in[0,\bar{s}-1]\}$ forms a basis of $\mathbb{K}$ over $F_{i^*}$. Therefore, for $j\in[u-\e,u-1]$, rack $i^*$ downloads
$$b'=|\{\tr_{\mathbb{K}/F_{i^*}}(\gamma_{w}e_{i,j}):i\in D,w\in[p_{i^*}]\}|=\bar{d}p_{i^*}$$ symbols over $ F_{i^*}$ from the relayers of helper racks, which can be regarded as downloading $\frac{\bar{d}t}{\bar{d}-k'+1}$ sub-symbols over $\fq$. Then one can repair $\e$ coefficients
\begin{eqnarray}\label{Eq_Co_Chen}
\{e_{i^*,u-\e},\cdots,e_{i^*,u-1}\}
\end{eqnarray}
of $f_{i^*}(x)$, thereby the failed symbols can be repaired from (\ref{Eq_Co_Chen})  and the surviving nodes within the rack.
\end{enumerate}

\begin{corollary}
The rack-aware RS codes given in \cite{CB19} with $1\leq \e\leq u$ failures located on the same rack can be repaired with bandwidth $\frac{\e\bar{d}t}{\bar{d}-\lceil\frac{k}{u}\rceil+1}$ sub-symbols.
\end{corollary}

 {Notably, for the case of $u|k$, our repair scheme achieves the same optimal system bandwidth as that in \cite{WZL+23}, meeting the cut-set bound on the repair bandwidth for repairing multiple failures in the same rack as given in \cite{CB19}.}

\subsection{New Constructions of Rack-Aware Reed-Solomon Codes}

In this subsection, we employ the good polynomial proposed by Tamo and Barg~\cite{TB14} to organize nodes of RS codes, leading to some new constructions of rack-aware RS codes with efficient repair bandwidth.
To begin with, we introduce three classes of good polynomials in \cite{TB14}.

Suppose that $u=mq^a$, where $m|(q^t-1)$, $0\leq a\leq t$ {and $t$ is the extension degree of $\fqt$ over $\fq$}. Let $E$ be a multiplicative subgroup of $\fqt^*$ with the order $m$ and $U,V$ be two additive subgroups of $\fqt$ with the size $q^a$. There are three classes of good polynomials over $\fqt$.

\begin{itemize}
\item When $a=0$ and $m>1$,
\begin{equation}\label{eq: gp1}h(x)=x^m\end{equation}
is a good polynomial satisfying $h(x)=h(g)$ for any $x\in g{E}$, where $g\in\fqt^*$.
\item When $a>0$ and $m=1$,
\begin{equation}\label{eq: gp2}h(x)=\sum_{i=0}^{a}\theta_{i}x^{q^i}\end{equation}
is a good polynomial satisfying $h(x)=h(b)$ for any $x\in b+{U}$, where $b\in\fqt$, $\theta_{i}\in\fqt$ with $\theta_{0}\neq0, \theta_{a}\neq0$.
\item When $a>0$ and $m>1$, $e|t$ and $m|(q^e-1)$,
$$h(x)=\Pi_{i=1}^m\Pi_{v\in V}(x+v+\beta_{i})$$
is a good polynomial satisfying $h(x)=h(\alpha)$ for any $x\in\cup_{i=1}^{m}(\alpha\beta_{i}+V)$, where $\alpha\in\fqt^*\backslash V$, $V$ is an additive subgroup of $\fqt$ with the size $q^a$  that is closed under  the multiplication by the field $\mathbb{F}_{q^e}$, and $\beta_{1},\cdots,\beta_{m}$ are the $m$-th degree roots of unity in $\fqt$. Assume that $V$ contains the element $1$, then the above polynomial can be rewritten as
\begin{equation}\label{eq: gp3}h(x)=(\sum_{i=0}^{a/e}\theta_{i}x^{q^{ei}})^m,\end{equation}
where $\theta_{i}\in\fqt$ satisfy $\sum_{i=0}^{a/e}\theta_{i}=0$, $\theta_{0}\neq0$, and $ \theta_{a/e}\neq0.$
\end{itemize}

By Construction \ref{cons}, we employ the good polynomials in \eqref{eq: gp1}-\eqref{eq: gp3} to organize the nodes of RS codes into different racks.  Consequently,  we can transform the repair of the resultant rack-oriented RS code into repairing a short RS code under the homogeneous storage model, which has been pursued in literature {\cite{BBD+22, TY17, XZ23,YB16, DDK+17, GW17, LWJ19,ZZ19}}. With the help of those known repair schemes and Corollary \ref{coro:rack_b}, we obtain some rack-aware RS codes with efficient repair properties. {The constructions supporting failures within a single rack are presented in Table~II.}
In this table,  we provide explicit formulations of good polynomials denoted as $h(x)$ for each construction, where the fourth class is the rack-aware RS code shown in Example $1$.  As we mentioned before,  Steps $1$ and $3$ of repairing the resulting rack-oriented RS code are generally identical. The distinction lies in the explicit repair scheme adopted in Step $2$, which is feasible to utilize all the known ones for solving single-node failure under the homogeneous storage model. {For the explicit scheme used in Step $2$, refer to the references provided in the last column of the table, which varies on a case-by-case basis and is therefore omitted here for simplicity.  Especially, since the homogeneous RS code given in \cite{TY17} is optimal for repairing a single node, the constructions given in rows $2,4,6$ of Table I have the optimal repair bandwidth in the case of single-rack failure when $u|k$. }

\begin{table*}[]\label{table: constructions}
\centering
\caption{New Rack-Aware Reed-Solomon Codes with Single-Rack Repair Property}
 \vspace{1mm}
\label{tab:my-table}
\resizebox{0.9\textwidth}{!}
{%
\begin{tabular}{@{}ccccccc@{}}
\toprule
$h(x)$&
  Evaluation Points (for $i$-th rack) &
  Sub-Packetization &
  Rack Size &
  Bandwidth (per failure) &
  Conditions&Repair Schemes for Step 2 \\ \midrule
    \vspace{1mm}
\multirow{2}{*}{$x^m$} &
  $A_{i}=\alpha^{\bar{r}^{i-1}}E$ &
  $t=\bar{r}^{\bar{n}}$ &
  $u=m$&
  $<\frac{(\bar{n}+1)t}{\bar{n}-k'}$ &
  $\fq(\alpha)=\fqt, \forall\delta|t, \delta<t, ord(\alpha^{m})\nmid (q^{\delta}-1)$ &\cite{YB16} \\
  \vspace{1mm}
 &
  $A_{i}=\alpha_{i}E$ &
  $t=(\bar{d}-k'+1)\Pi_{i=1}^{\bar{n}}p_{i}$ &
  $u=m$ &
  $\frac{\bar{d}t}{\bar{d}-k'+1}$ &
  $\fq(\alpha_{i})=\f_{q^{p_{i}}}, ord(\alpha_{i}^m)\nmid (q-1)$&\cite{TY17} \\
  \hline
    \vspace{1mm}
\multirow{2}{*}{$\tr_{\fqt/ \mathbb{F}_{q^{t-a}}}(x)$} &
  $A_{i}=b_{i}+U$ &
  $t=\bar{r}^{\bar{n}}+a$ &
  $u=q^a$ &
  $<\frac{(\bar{n}+1)t}{\bar{n}-k'}$ &
  $h(b_{i})=\alpha^{\bar{r}^{i-1}},\fq(\alpha)=\f_{q^{t-a}}, (t-a)|t$& \cite{YB16}\\
    \vspace{1mm}
 &
  $A_{i}=b_{i}+U$ &
  $t=(\bar{d}-k'+1)\Pi_{i=1}^{\bar{n}}p_{i}+a$ &
{$u=q^a$} &
  $\frac{\bar{d}t}{\bar{d}-k'+1}$ &
  $h(a_{i})=\alpha_{i},\fq(\alpha_{i})=\f_{q^{p_{i}}}, (t-a)|t$ &\cite{TY17}\\
    \hline
      \vspace{1mm}
\multirow{2}{*}{$(\tr_{\fqt/ \mathbb{F}_{q^{t-a}}}(x))^m$} &
 {$A_{i}=\cup_{j=1}^{m}a_{i}\beta_{j}+V$} &
{$t=\bar{r}^{\bar{n}}+a$ }&
{$u=mq^a$} &
{$<\frac{(\bar{n}+1)t}{\bar{n}-k'}$} &
  $\tr_{\fqt/ \mathbb{F}_{q^{t-a}}}(a_{i})=\alpha^{\bar{r}^{i-1}},\fq(\alpha)=\f_{q^{t-a}},m|(q^{t-a}-1),\forall\delta|t, \delta<t, ord(\alpha^{m})\nmid (q^{\delta}-1), (t-a)|t,p|\frac{t}{t-a}$&\cite{YB16} \\
 \vspace{1.5mm}
 &
   \multirow{2}{*}{$A_{i}=\cup_{j=1}^{m}a_{i}\beta_{j}+V$ }&
   \multirow{2}{*}{$t=(\bar{d}-k'+1)\Pi_{i=1}^{\bar{n}}p_{i}+a$} &
   \multirow{2}{*}{$u=mq^a$ }&
   \multirow{2}{*}{$\frac{\bar{d}t}{\bar{d}-k'+1}$ }&
  $\tr_{\fqt/ \mathbb{F}_{q^{t-a}}}(a_{i})=\alpha_{i},\fq(\alpha_{i})=\f_{q^{p_{i}}},m|(q^{t-a}-1), ord(\alpha_{i}^m)\nmid (q-1), (t-a)|t,p|\frac{t}{t-a}$&\cite{TY17}\\
       \hline
      \vspace{1mm}
$\beta\tr_{\fqt/ \mathbb{F}_{q^{t-a}}}(x)$ &
  $A_{i}=b_{i}+U$ &
  $t>0$ &
  $u=q^a$ &
  $\frac{t}{t-a}(\bar{n}-1)(t-a-l)$ &
  $\bar{n}=q^{t-a},\bar{n}-k'\geq q^l,1\leq l <t-a, (t-a)|t$ &\cite{LWJ19}\\
    \hline
      \vspace{1mm}
$L_{U}(x)$ &
  $A_{i}=b_{i}+U$ &
  $t>0$ &
  $u=q^a$ &
  $(\bar{n}-1)w$ &
  $\bar{n}=q^{t-a},\bar{n}-k'=q^l,1\leq l<t-a, tl\geq (t-a)(t-w)$&\cite{BBD+22,DDK+17}\\
  \hline
\multicolumn{7}{c}{ Denote by $k'=\lceil\frac{k}{u}\rceil$, $\bar{r}=\bar{n}-k'$ and {$k' \leq \bar{d}\leq \bar{n}-1$}. Let $p$ be the characteristic of $\fq$ and $p_{1},p_{2},\cdots,p_{\bar{n}}$ be coprime satisfying $p_{i}\equiv 1
\bmod(\bar{d}-k'+1)$, where $i\in[\bar{n}]$.  $\beta$ is a nonzero element in $\fqt$. }\\
\multicolumn{7}{c}{The repair schemes for the homogeneous storage model used in Step $2$ refer to the references presented in the last column.}\\
\end{tabular}%
}
\end{table*}

{Similarly, the rack-aware RS codes that can tolerate multiple-rack failures are presented in Table III.
In this case, we assume that failures are distributed among $1<h\leq \bar{n}-\lceil\frac{k}{u}\rceil$ failed racks, each containing $\e_1,\cdots,\e_h$ failed nodes respectively, where $h$ represents the number of failed racks. The first three classes of constructions repair multiple racks in a centralized manner, while the others require cooperation between racks.

Due to the limited space, we only list the bandwidth for the case where $\e_1=\e_2=\cdots=\e_h$ in this table. For the non-regular cases, readers can easily determine the bandwidth according to our repair scheme.
For example, we use the case of two failed racks, employing the first class of codes in Table III to illustrate the general case of our repair framework, i.e., $1\leq\e_1\leq\e_2\leq u$.

Let $k=u\bar{k}+v,$ $0\leq v\leq u-1$ and $q$ be a power of a prime. Let $p_{i}$ be distinct primes such that $p_{i}\equiv1$ mod $\bar{s}$, where $\bar{s}=\bar{s}_1\bar{s}_2=(\bar{d}-k'+1)(\bar{d}-k'+2)$ and $p_{i}>u$ for $i\in[\bar{n}]$, where $k'=\lceil\frac{k}{u}\rceil$. We take the smallest $\bar{n}$ primes. For any $i\in[\bar{n}]$, let $\lambda_{i}$ be an element of degree $p_{i}$ over $\fq$, namely, $[\fq(\lambda_{i}),\fq]=p_{i}$. Define the finite field $F:=\fq(\lambda_{j}:j\in[\bar{n}]\backslash\{i_1,i_2\})$ and $\mathbb{F}:=\fq(\lambda_{1},\cdots,\lambda_{\bar{n}})$, then we can deduce that $[\f:F]=p_{i_1}p_{i_2}$.
Let $\mathbb{K}$ be an extension of $\mathbb{F}$ with degree $\bar{s}$. It is known that $\mathbb{K}$ is also an extension of $\fq$ with degree $t=[\mathbb{K}:\fq]=\bar{s}\Pi_{i=1}^{\bar{n}}p_{i}$.

Consider an  RS code $\mathcal{C}={\rm{RS}}(n,k,A)$ over $
\fqt$ with
$A=\cup_{i=1}^{\bar{n}}\{\lambda_{ij}:j\in[u]\},$
where $\lambda_{ij}=\lambda_{i}\lambda^{j-1}$ and $\lambda\in\fq$ is an element of order $u$.
Since $\lambda_{i}$ is a generator of $\mathbb{F}_{q^{p_{i}}}$ over $\fq$, and $p_{i}$ is a prime with $p_{i}>u$, the element $\lambda_{i}^u$ satisfies $\fq(\lambda_{i}^u)=\mathbb{F}_{q^{p_{i}}}$ as well .

\begin{enumerate}
\item[1)] Let $$h(x)=x^u,$$ then $h(A_{i})=\lambda_{i}^u,$ where $A_{i}=\{\lambda_{ij}:j\in[u]\}$ and $i\in[\bar{n}]$. Then, place the $u$ nodes corresponding to $A_{i}$ into $i$-th rack, i.e.
\begin{eqnarray*}
(f(\lambda_{i1}),f(\lambda_{i2}),\cdots,f(\lambda_{iu}))
\end{eqnarray*} are stored in $i$-th rack, $i\in[\bar{n}]$.

\item[2)]  Suppose that $1\leq\e_1\leq\e_2\leq u$ failures occur in the racks $i_1$ and $i_2$ respectively.  $D\subset [\bar{n}]\backslash\{i_1,i_2\}$ represents the set of helper racks of size $\bar{d}$. According to the repair procedure given in Section \ref{sec: cons_1}, the repair problem is transformed to {repairing} two erasures of an $[\bar{n},k']$ RS code defined by $\bar{A}=\{\lambda_{i}^u:i\in[\bar{n}]\}$ under the homogeneous storage model, which satisfied optimal repair property according to \cite{TY19}.
Then, by $\cite{TY19}$, there exist two sets $S_{i_1}$ and $S_{i_2}$ with size $\bar{s}_2p_{i_1}p_{i_2}$ and $\bar{s}_1p_{i_1}p_{i_2}$
respectively, where
\begin{equation*}
{\rm dim}_{F}({\rm Span}_{F}(S_{i_1})\cap {\rm Span}_{F}(S_{i_2}))=p_{i_1}p_{i_2}.
\end{equation*}
For any $j\in[u-\e_1,u-1]$,
$e_{i_1,j}$ can be calculated from the set of symbols $\{\tr_{\mathbb{K}/F}(\gamma e_{i,j}): \gamma\in S_{i_1}, i\in D\}$, and then $e_{i_2,j}$ can be calculated from the recovered coefficient $e_{i_1,j}$ and the set of symbols $\{\tr_{\mathbb{K}/F}(\gamma e_{i,j}): \gamma\in S_{i_2}, i\in D\}$.
Therefore,
the repair center downloads
$$b'_1=|\{\tr_{\mathbb{K}/F}(\gamma e_{i,j}):i\in D,\gamma\in S_{i_1}\cup S_{i_2}\}|=2\bar{s}_1p_{i_1}p_{i_2}$$ symbols over $ F$ from the relayers of helper racks, which can be regarded as downloading $\frac{2\bar{d}t}{\bar{d}-k'+2}$ sub-symbols over $\fq$. Then one can repair $\e_1$ coefficients
\begin{eqnarray}\label{Eq_Co_m}
\{e_{i_j,u-\e_1},\cdots,e_{i_j,u-1}\},\,\,j=1,2
\end{eqnarray}
of $f_{i_j}(x)$ with bandwidth
\begin{equation*}
b_1=\e_1 b'_1=\frac{2\e_1\bar{d}t}{\bar{d}-\lceil\frac{k}{u}\rceil+2}
\end{equation*}
 sub-symbols.

 For $j\in[u-\e_1-\e_2,u-\e_1-1]$, the repair of
 \begin{equation}\label{Eq_Co_m_2}
 \{e_{i_2,u-\e_1-\e_2},\cdots, e_{i_2,u-\e_1-1}\}
 \end{equation}
 degenerates into repairing a single node of an $[\bar{n},k']$ RS code, which also exists an optimal repair scheme with bandwidth
 \begin{equation*}
 b_2=\frac{(\e_2-\e_1)\bar{d}t}{\bar{d}-\lceil\frac{k}{u}\rceil+1}
 \end{equation*}
sub-symbols.

Therefore, the failed symbols can be repaired from \eqref{Eq_Co_m}, \eqref{Eq_Co_m_2}  and the surviving nodes within each failed rack. The total repair bandwidth is
\begin{equation*}
b=b_1+b_2=\frac{2\e_1\bar{d}t}{\bar{d}-\lceil\frac{k}{u}\rceil+2}+\frac{(\e_2-\e_1)\bar{d}t}{\bar{d}-\lceil\frac{k}{u}\rceil+1}.
\end{equation*}
\end{enumerate}

\begin{table*}[]\label{table: multi_constructions}
\centering
\caption{New Rack-Aware Reed-Solomon Codes with Multiple-Rack Repair Property}
 \vspace{1mm}
\label{tab:my-table}
\resizebox{0.9\textwidth}{!}
{
\begin{tabular}{@{}ccccccc@{}}
\toprule
$h(x)$&
  Evaluation Points (for $i$-th rack) &
  Sub-Packetization &
  Rack Size &
  Bandwidth (per failure) &
  Conditions & Repair Schemes for Step 2\\ \midrule
    \vspace{1mm}
{$x^m$} &
  $A_{i}=\alpha_{i}E$ &
  $t=\bar{s}\Pi_{i=1}^{\bar{n}}p_{i}$ &
  $u=m$ &
  $\frac{h\bar{d}t}{\bar{d}-k'+h}$ &
  $ \fq(\alpha_{i})=\f_{q^{p_{i}}}, ord(\alpha_{i}^m)\nmid (q-1)$ &\cite{TY19}\\
  \hline
    \vspace{1mm}
{$\tr_{\fqt/ \mathbb{F}_{q^{t-a}}}(x)$} &
  $A_{i}=b_{i}+U$ &
  $t=\bar{s}\Pi_{i=1}^{\bar{n}}p_{i}+a$ &
{$u=q^a$} &
  $\frac{h\bar{d}t}{\bar{d}-k'+h}$ &
  $h(a_{i})=\alpha_{i},\fq(\alpha_{i})=\f_{q^{p_{i}}}, (t-a)|t$&\cite{TY19} \\
    \hline
      \vspace{1mm}
{$(\tr_{\fqt/ \mathbb{F}_{q^{t-a}}}(x))^m$} &
  {$A_{i}=\cup_{j=1}^{m}a_{i}\beta_{j}+V$ }&
  {$t=\bar{s}\Pi_{i=1}^{\bar{n}}p_{i}+a$} &
  {$u=mq^a$ }&
{$\frac{h\bar{d}t}{\bar{d}-k'+h}$ }&
  $\tr_{\fqt/ \mathbb{F}_{q^{t-a}}}(a_{i})=\alpha_{i},\fq(\alpha_{i})=\f_{q^{p_{i}}},m|(q^{t-a}-1), ord(\alpha_{i}^m)\nmid (q-1), (t-a)|t,p|\frac{t}{t-a}$&\cite{TY19}\\

    \hline
      \vspace{1mm}
\multirow{9}{*}{$\beta\tr_{\fqt/ \mathbb{F}_{q^{t-a}}}(x)$} &
  \multirow{9}{*}{$A_{i}=b_{i}+U$} &
  \multirow{9}{*}{$t>0$ }&
  \multirow{9}{*}{$u=q^a$} &
  \multirow{9}{*}{$\frac{ht}{t-a}(\bar{n}-1)(t-a-l)$} &
  $h=2,\bar{n}=q^{t-a},\bar{n}-k'\geq q^{l},l=t-a-1, (t-a)|t$ &\\
 \vspace{1mm}
 &
  &
 &
  &
 &

  $h=3,-,\{\frac{y_{i}-y_j}{y_{i}-y_{w}}:i\neq j\neq w\in\{i^*_j:j\in[3]\}\}\subset \fq^*$&\cite{ZZ19}\\
   \vspace{1mm}
 &
  &
 &
  &
 &

  $h=3,-,(t-a)^2\neq 1 \bmod p$,$\{\frac{y_{i}-y_j}{y_{i}-y_{w}}:i\neq j\neq w\in\{i^*_j:j\in[3]\}\}\subset K$&\\

 \vspace{1mm}
 &
  &
 &
  &
 &

  $h=3,-,t-a>3$&\\

      \vspace{1mm}
 &
  &
 &
  &
 &

  $h=2,\bar{n}=q^{t-a},\bar{n}-k'\geq q^l, l >(t-a)/2, (t-a)|t$& \\

      \vspace{1mm}
 &
  &
 &
  &
 &

  $h=2,\bar{n}=q^{t-a},\bar{n}-k'\geq q^l, l |(t-a), (t-a)|t$&\multirow{2}{*}{\cite{XZ23}} \\

      \vspace{1mm}
 &
  &
 &
  &
 &

  $h=2,\bar{n}=q^{t-a},\bar{n}-k'\geq q^l,t-a=p^\delta e,\delta>0,p\nmid e,e\leq l<t-a, (t-a)|t$ & \\
  \vspace{1mm}
 &
  &
 &
  &
 &

  $h=2,\bar{n}=q^{t-a},\bar{n}-k'\geq q^l, t-a=2l+1,p\nmid 2l, (t-a)|t$ &\\
\vspace{1mm}
 &
  &
 &
  &
 &

  $h=2,\bar{n}=q^{t-a},\bar{n}-k'\geq q^l, t-a=2l+2,p\nmid l, (t-a)|t$& \\

   \hline
\multicolumn{7}{c}{ Denote by $k'=\lceil\frac{k}{u}\rceil$, $\bar{r}=\bar{n}-k'$ and $k'\leq\bar{d}\leq \bar{n}-h$, where $h$ represents the number of failed racks satisfying $1<h\leq \bar{r}$. Let $p$ be the characteristic of $\fq$ and $p_{1},p_{2},\cdots,p_{\bar{n}}$ be coprime satisfying $p_{i}\equiv 1
\bmod \bar{s}$, where $\bar{s}=\Pi_{j=1}^{h}(\bar{d}-k'+j)$ and $i\in[\bar{n}]$.   }\\

\multicolumn{7}{c}{Let $\beta$ be a nonzero element in $\fqt$, and $K={\rm{Ker}}(\tr_{\f_{q^{t-a}}/\fq})$. For simplicity, we use ``-" to represent the conditions ``$\bar{n}=q^{t-a},\bar{n}-k'\geq q^{l},l=t-a-1, (t-a)|t$". }\\

\multicolumn{7}{c}{The repair schemes for the homogeneous storage model used in Step $2$ refer to the references presented in the last column.}\\

\end{tabular}%
}
\end{table*}

}

\section{Conclusion}\label{sec: conclusion}
{In this paper, we employed good polynomials to arrange nodes of RS codes on racks and propose a generic repair framework for handling multiple failures within racks.
Building upon this, the problem of multiple-node recovery in RS codes under the rack-aware storage model is transformed into the repair problem of RS codes under the homogeneous storage model. This transformation leverages the benefits of both rack-aware and homogeneous RS codes, providing a solution for multiple-node repair in the rack-aware storage model.

 In this way, we generalized the existing constructions to support multiple-node failures and further designed several new constructions using different good polynomials. These proposed codes can be repaired by means of the known schemes  {of RS code under the homogeneous storage model.
}

\end{document}